\pdfoutput=1
\documentclass[lettersize,draftclsnofoot,onecolumn]{IEEEtran}
\usepackage{amsmath,amsfonts}
\usepackage{algorithmic}
\usepackage{array}
\usepackage[caption=false,font=normalsize,labelfont=sf,textfont=sf]{subfig}
\usepackage{textcomp}
\usepackage{multirow}
\usepackage{xcolor}
\usepackage[normalem]{ulem}
\usepackage{stfloats}
\usepackage{float}
\usepackage{amssymb}
\usepackage[hidelinks]{hyperref}
\newtheorem{assumption}{Assumption}

\newtheorem{theorem}{Theorem}
\newtheorem{lemma}{Lemma}
\newtheorem{proposition}{Proposition}
\newtheorem{remark}{Remark}

\usepackage{verbatim}
\usepackage{xcolor}
\usepackage{graphicx}
\usepackage{subfig}
\def\BibTeX{{\rm B\kern-.05em{\sc i\kern-.025em b}\kern-.08em
		T\kern-.1667em\lower.7ex\hbox{E}\kern-.125emX}}
\usepackage{balance}
\begin{document}
	\title{A Finite-Blocklength Analysis for ORBGRAND}
	\author{\IEEEauthorblockN{Zhuang Li and Wenyi Zhang},~\IEEEmembership{Senior Member,~IEEE}
		\thanks{The authors are with Department of Electronic Engineering and Information Science, University of Science and Technology of China, Hefei, China (wenyizha@ustc.edu.cn).}
	}
	
	\maketitle
	
	\begin{abstract}
		
		Within the Guessing Random Additive Noise Decoding (GRAND) family, ordered reliability bits GRAND (ORBGRAND) has received considerable attention for its hardware-friendly exploitation
		of soft information. Existing information-theoretic results
		for ORBGRAND are asymptotic in blocklength and do not quantify
		its performance at short-to-moderate blocklengths. This paper
		develops a finite-blocklength analysis for ORBGRAND over
		general binary-input memoryless channels, addressing the key challenge that the rank-induced decoding metrics are coupled across symbols.  
		We first specialize the general random-coding union (RCU) bound for
		mismatched decoding to ORBGRAND, yielding
		an ORBGRAND-specific form termed the ORB-RCU bound. We then characterize the two codeword decoding metrics governing this bound:
		the transmitted-codeword decoding metric is represented, up to a deterministic $O(n^{-1})$ term, by a nondegenerate second-order $U$-statistic and admits a uniform Gaussian approximation via the Hoeffding projection and a Berry-Esseen theorem for $U$-statistics, whereas the competing-codeword decoding metric is reduced to a weighted sum of independent and identically distributed Bernoulli random variables and characterized through a uniform strong large-deviation expansion.
		Combining these ingredients, we establish a uniform Gaussian approximation for the ORB-RCU bound, which in turn yields a third-order achievable rate expansion for ORBGRAND. The first-order term is shown to coincide with the previously established ORBGRAND achievable rate, while the second-order term defines an ORBGRAND dispersion, and the third-order correction is given by $\frac{\ln n}{2n}$. Numerical results for the BPSK-modulated additive white Gaussian noise channel validate the tightness of ORB-RCU relative to the maximum-likelihood based RCU benchmark and
		show that the resulting approximation accurately predicts the finite-blocklength performance in the operating regime of practical interest.
		
	\end{abstract}
	
	
	\section{Introduction}
	
	Guessing Random Additive Noise Decoding (GRAND) \cite{duffy2019capacity,riaz2022universal} has attracted significant attention in recent years as a decoding paradigm that shifts the search from codewords to noise—more precisely, to error patterns (EPs). Given a received block, GRAND queries a sequence of candidate EPs and, for each query, performs a codeword-membership check; decoding terminates once removing the hypothesized EP yields a valid codeword (for linear block codes, this can be implemented efficiently via syndrome checks). This architecture is broadly applicable to arbitrary block codes and is particularly well suited to high-rate short-blocklength regimes \cite{yue2023efficient,wang25arxiv}, making it a promising candidate for meeting the stringent latency and reliability requirements of ultra-reliable low-latency communication (URLLC).
	
	Among GRAND variants for soft-output channels, soft GRAND (SGRAND) \cite{solomon2020soft} orders EPs using the magnitudes of the channel-output log-likelihood ratios (LLRs) and, when used without a limit on the maximum number of queries, yields decisions equivalent to maximum-likelihood (ML) decoding \cite{liu2022orbgrand}\cite[Ch.~10]{Lin2004ErrorCC}. However, its reliance on exact LLR values typically necessitates on-the-fly EP construction for each received block, which complicates hardware implementation \cite{solomon2020soft,wan2026parallelization}. In contrast, ordered reliability bits GRAND (ORBGRAND) \cite{duffy2022ordered} generates EPs using only the ranking of LLR magnitudes. This rank-based design enables efficient and hardware-friendly EP generation \cite{abbas2022high,condo2022fixed}. ORBGRAND has attracted broad research interest, ranging from information-theoretic studies \cite{liu2022orbgrand,li2026orbgrand} to algorithmic refinements that enhance decoding performance and narrow the gap to ML decoding \cite{condo2021high,liu2022orbgrand,duffy2022ordered,condo2022fixed,wan2024approaching,wan2025fine}.
	
	Existing information-theoretic analyses establish that ORBGRAND is nearly capacity-achieving for the additive white Gaussian noise (AWGN) channel with antipodal inputs \cite{liu2022orbgrand}, and can be made exactly capacity-achieving for general binary-input memoryless channels via rank companding \cite{li2026orbgrand}. These results, however, are asymptotic in blocklength and therefore do not quantify performance in the short-to-moderate blocklength regime where ORBGRAND is the most relevant. Complementary finite-length achievability and converse bounds,
	together with a simple and accurate approximation, were
	developed within a framework based on guesswork and the GRAND
	family of decoding algorithms in \cite{mariona2024finite}. Guessing-based decoders with abandonment for discrete memoryless channels (DMCs) under constant composition codebooks was further studied in \cite{tan2025ensemble}, where a second-order analysis was obtained using a type-based decoding rule. Unlike the decoding rules considered in these works, ORBGRAND orders EPs according to the global, output-dependent ranking of LLR magnitudes across symbol positions, which induces cross-symbol coupling and requires a distinct finite-blocklength analysis. Such an analysis enables a quantitative assessment of the finite-blocklength gap to ML decoding and the penalty induced by rank-based EP ordering. The resulting finite-blocklength approximations serve as tractable surrogates for system-level performance evaluation, facilitating rapid exploration of rate-reliability-latency tradeoffs in URLLC settings.

	In the finite-blocklength regime, the fundamental tradeoff among rate, error probability, and blocklength is characterized by the nonasymptotic framework developed in
	\cite{polyanskiy2010channel,polyanskiy2010channel2}, which
	provides sharp converse and achievability benchmarks,
	including the meta-converse and the
	random-coding union (RCU) bounds, together with an accurate
	normal approximation under ML decoding.
	In particular, let $M^\star(n,\epsilon)$ denote the maximum
	codebook size at blocklength $n$ with average error probability
	not exceeding $\epsilon\in(0,1)$.
	The corresponding second-order expansion is
	\begin{equation}
		\log M^\star(n,\epsilon)
		=
		nC-\sqrt{nV}\,Q^{-1}(\epsilon)+O(\log n),
		\label{eq:PPV_second_order}
	\end{equation}
	where $C$ and $V$ are the channel capacity and dispersion,
	respectively, and $Q^{-1}(\cdot)$ is the inverse of the Gaussian
	tail function.\footnote{The Gaussian tail function is defined
	as $Q(x)\triangleq\int_x^{\infty} \frac{1}{\sqrt{2\pi}}e^{-t^2/2}\,\mathrm{d}t$.}
	This expansion refines the Shannon channel coding theorem by
	quantifying the finite-blocklength backoff from $C$, with $V$
	governing the $\sqrt n$-scale penalty at reliability level $\epsilon$ \cite{strassen1962asymptotische,hayashi2009information,
	altuug2014moderate}. Beyond the second-order term, it is worth noting that the
	mechanism producing the logarithmic third-order correction is
	well documented in the finite-blocklength literature. It was
	observed in \cite{martinez2011saddlepoint} that loosening the
	RCU bound through Markov's inequality may forfeit $O(\log n)$ term, whereas sharply characterizing the pairwise
	error probability with an explicit prefactor of order
	$n^{-1/2}$ restores an additional $\frac{1}{2}\log n$ term in
	the rate expansion. Third-order refinements were
	subsequently established for DMCs
	\cite{tomamichel2013tight}, while more detailed characterizations were developed through saddlepoint approximations \cite{scarlett2014saddlepoint, font2018saddlepoint} and
	extended to mismatched decoding in
	\cite{scarlett2014mismatched}. Further third-order results were
	obtained for the AWGN channel \cite{tan2015third}, and strong
	large-deviation analysis combined with refined central limit
	asymptotics made the $\frac{1}{2}\log n$ term explicit for
	memoryless channels while controlling the remaining term
	uniformly \cite{moulin2017log}. Combinatorial central limit theorems have also been employed for constant composition random coding over multiple-access channels \cite{scarlett2014second}.
	
	Nevertheless, a common structural assumption underlying most of the aforementioned
	finite-blocklength analyses is that the decoding metric for each position depends only on the channel variables in that position. In contrast, ORBGRAND is driven by reliability ordering, and
	its induced decoding metrics are rank-based and
	coupled across symbols, even when the channel is memoryless. Hence, the existing second- and third-order analyses cannot be directly applied to ORBGRAND, and a dedicated finite-blocklength analysis is required.

	In this paper, we address this challenge
	by conducting a refined analysis
	that handles
	the coupling across symbols. We start with a
	finite-blocklength upper bound on the ensemble-average decoding
	error probability of ORBGRAND, termed the ORB-RCU bound, 
	which is obtained by specializing the general mismatched RCU bound \cite{martinez2011saddlepoint, scarlett2014saddlepoint} to ORBGRAND. This bound is governed by two codeword decoding metrics: the transmitted-codeword decoding metric and the competing-codeword decoding metric. For the transmitted-codeword decoding metric, we
	represent it by a nondegenerate second-order $U$-statistic up to a deterministic $O(n^{-1})$ term and, through the Hoeffding projection, establish a uniform Berry-Esseen
	approximation for its $n^{-1/2}$-scale fluctuations; for the competing-codeword decoding metric, we show that it is distributionally equivalent to a weighted sum of independent and identically distributed (i.i.d.)
	Bernoulli random variables and derive from first principles a uniform strong large-deviation
	expansion for its tail probability using exponential tilting,
	uniform finite-blocklength saddlepoint expansions, and a
	lattice local central limit argument. Although the derivation
	is motivated by classical strong large-deviation analyses
	\cite{bahadur1960deviations,chaganty1993strong,joutard2013strong},
	the existing results do not directly cover the weighted lattice
	sum arising here or provide the uniform expansion
	required for the subsequent analysis.
	Combining these two ingredients through a
	threshold relaxation yields the following second-order
	achievable rate expansion, where $R_{\mathrm{ORB}}^{\star}(n,\epsilon)$ denotes $\frac{1}{n}\ln M_{\mathrm{ORB}}^{\star}(n,\epsilon)$, with $M_{\mathrm{ORB}}^{\star}(n,\epsilon)$ being the largest integer $M$ for which there exists an $(n,M)$ code whose average error probability under ORBGRAND does not exceed $\epsilon$:
	\begin{align}
		R_{\mathrm{ORB}}^{\star}(n,\epsilon)
		\ge
		I_{\mathrm{ORB}}
		-
		\sqrt{\frac{V_{\mathrm{ORB}}}{n}}\,
		Q^{-1}(\epsilon)
		+
		O\left(\frac{1}{n}\right),
		\label{eq:second-order-expan}
	\end{align}
	where $I_{\mathrm{ORB}}$ is the
	ORBGRAND achievable rate previously established in
	\cite{liu2022orbgrand,li2026orbgrand}, and
	$V_{\mathrm{ORB}}$ is the dispersion characterized in this
	work. To obtain a more refined expansion, we then analyze the pairwise probability appearing within the ORB-RCU bound directly, without introducing an intermediate threshold relaxation. This direct analysis retains the explicit
	$n^{-1/2}$ prefactor in the strong large-deviation expansion
	of the competing-codeword term and yields the following
	third-order achievable rate expansion:
	\begin{align}
		R_{\mathrm{ORB}}^{\star}(n,\epsilon)
		\ge
		I_{\mathrm{ORB}}
		-
		\sqrt{\frac{V_{\mathrm{ORB}}}{n}}\,
		Q^{-1}(\epsilon)
		+
		\frac{\ln n}{2n}
		+
		O\left(\frac{1}{n}\right).
		\label{eq:expan}
	\end{align}
	This logarithmic correction is conceptually analogous to that arising in \cite{martinez2011saddlepoint}; see also
	\cite{scarlett2014saddlepoint,font2018saddlepoint,
	scarlett2014mismatched} for further details. Its derivation here,
	however, requires an ORBGRAND-specific treatment because the
	decoding metrics are rank-based and coupled across symbols.

	Numerical results for the AWGN channel
	corroborate the analysis from several perspectives. The
	achievable rate and block error probability comparisons over a
	range of blocklengths show that ORBGRAND remains fairly close to
	ML decoding, with only a modest finite-blocklength penalty. The
	comparison between $V_{\mathrm{ORB}}$ and the corresponding
	channel dispersion $V$ further quantifies this penalty and provides
	additional insight into the finite-blocklength gap between
	ORBGRAND and ML decoding. Moreover, the third-order approximation
	obtained by retaining the first three terms in
	\eqref{eq:expan} accurately tracks the ORB-RCU benchmark even at
	short blocklengths (e.g., $n = 100$), with the approximation becoming increasingly accurate as the blocklength grows.
	
	The remaining part of this paper is organized as follows.
	Section~\ref{sec:systemModel} introduces the channel model,
	reviews the basic 
	principle of ORBGRAND, and specifies the
	associated decoding rule.
	Section~\ref{sec:rcu} specializes the general mismatched RCU bound to ORBGRAND and characterizes the asymptotic behavior of the transmitted-codeword
	decoding metric.
	In particular, we represent the transmitted-codeword decoding metric as a nondegenerate second-order $U$-statistic up to a deterministic $O(n^{-1})$ term and establish a Berry-Esseen approximation.
	Section~\ref{sec:strong} develops a
	uniform strong large-deviation expansion for
	the cumulative distribution function (CDF) of the competing-codeword decoding metric.
	Section~\ref{sec:second_third_order} combines these two asymptotic ingredients to first establish a threshold-based second-order achievable rate expansion and then derive a refined third-order
	expansion through a direct analysis of the ORB-RCU bound. Section~\ref{sec:numerical} presents numerical results comparing
	the finite-blocklength bounds with the second- and third-order approximations. Section~\ref{sec:con} concludes the paper. Technical proofs are provided in the appendices.
	 
	\section{System Model}\label{sec:systemModel}
	
	Unless otherwise specified, we adopt the following notation
	throughout the paper. Random variables are typeset in sans-serif
	font, whereas their realizations are typeset in italic font.
	Underlined symbols are used to denote vectors. Accordingly,
	$\mathsf X$ and $x$ denote a random variable and one of its
	realizations, respectively, while
	$\underline{\mathsf X}$ and $\underline{x}$ denote a random
	vector and one of its realizations. The $i$-th component of
	$\underline{\mathsf X}$ is denoted by $\mathsf X_i$, with
	realization $x_i$. An unsubscripted single-letter random variable or random pair denotes a generic copy of its coordinate counterpart. For example, $(\Lambda,\mathsf E)$ has the same joint distribution as $(\Lambda_i,\mathsf E_i)$ for every $i$. A primed symbol denotes an independent copy; in particular,
	$(\Lambda',\mathsf E')$ is an independent copy of $(\Lambda,\mathsf E)$.
	
	\subsection{Channel Model}
	
	We consider a general binary-input memoryless channel with input alphabet $\{+1,-1\}$. Let the output probability density function be $q^+(y)$ when $x=+1$ and $q^-(y)$ when $x=-1$, without imposing any symmetry condition.
	
	A codebook of blocklength $n$ and rate $R$ nats per channel use contains $M=\lceil e^{nR}\rceil$ codewords. The codeword corresponding to message $m$ is $\underline{x}(m)=[x_1(m),\cdots,x_n(m)]$. In our random-coding analysis, the symbols $\{x_i(m)\}_{i=1}^n$ are i.i.d.\ and uniformly distributed on $\{+1,-1\}$, and the message is chosen uniformly from $\{1,\cdots,M\}$.
	
	Given the channel output vector $\underline{\mathsf{Y}}=[\mathsf{Y}_1,\cdots,\mathsf{Y}_n]$, define the LLRs
	\begin{equation}
		\mathsf{L}_i\triangleq\ln\frac{q^+(\mathsf{Y}_i)}{q^-(\mathsf{Y}_i)},\qquad i=1,\ldots,n,
	\end{equation}
	and the associated reliability vector $\underline{\Lambda}=[\Lambda_1,\cdots,\Lambda_n]$, where $\Lambda_i\triangleq|\mathsf{L}_i|$. For each $i\in\{1,\cdots,n\}$, let $\mathsf{R}_i$ be the rank of $\Lambda_i$ within $\underline{\Lambda}$, from $1$ (the smallest) to $n$ (the largest).

	Define the hard-decision error indicators as
	\begin{align}
		\mathsf E_i
		&\triangleq
		\mathbf 1\!(
		\operatorname{sgn}(\mathsf L_i)\cdot\mathsf X_i<0),
		\qquad i=1,\ldots,n,
		\label{eq:def_Ei}
	\end{align}
	where $\operatorname{sgn}(t)=1$ if $t\ge 0$ and $-1$ otherwise. Let $(\Lambda,\mathsf E)$ denote the corresponding generic
	single-letter pair, and define
	\begin{align}
		\Psi(t)
		&\triangleq
		\Pr[\Lambda\le t],
		\qquad t\ge0,
		\nonumber\\
		\mu
		&\triangleq
		\mathbb E[\Psi(\Lambda)\mathsf E].
		\label{eq:def_E_mu}
	\end{align}
	We impose the following mild regularity assumption throughout the paper.

	\begin{assumption}\label{assump:regularity}
		The reliability random variable $\Lambda$ is continuously distributed.
	\end{assumption}

	\subsection{ORBGRAND}
	
	Given a channel output vector $\underline{y}$, we form the LLRs $l_i = \ln\frac{q^+(y_i)}{q^-(y_i)}$, the reliability vector $\bigl[|l_1|,\cdots,|l_n|\bigr]$, and the corresponding hard-decision vector
	$\underline{x}_{\mathrm{hard}}=[\operatorname{sgn}(l_1),\cdots,\operatorname{sgn}(l_n)]$. The EPs queried by GRAND can be collected into a $\{\pm1\}$-valued matrix $P\in\{\pm1\}^{2^n\times n}$. In the $q$-th query, the $q$-th row $P_{q,:}$ of $P$ specifies which entries of $\underline{x}_{\mathrm{hard}}$ are flipped: if $P_{q,i}=-1$, the $i$-th entry of $\underline{x}_{\mathrm{hard}}$ is flipped, and otherwise it is left unchanged. Once the resulting vector is found to be a codeword, it is declared as the decoder output.\footnote{Since $2^n$ is typically an exceedingly large quantity, in practice we can truncate the matrix to keep only its first $Q<2^n$ rows, where $Q$ is the maximum number of queries permitted. If all rows of $P$ are exhausted without finding a codeword, a decoding failure is declared.}
	
	For ORBGRAND, the rows of $P$ are ordered such that the sum reliability $\sum_{i: P_{q,i}=-1} r_i$ is non-decreasing in $q$, where $r_i$ denotes the rank of $|l_i|$ (from $1$ for the smallest to $n$ for the largest), as defined in the previous subsection. Efficient constructions of the matrix $P$ are available; see, e.g., \cite{duffy2022ordered, condo2021high}. In practice, $P$ is often constructed in a recursive fashion, and truncated to its first $Q$ rows (maximum number of queries). In this paper, in order to analyze the fundamental performance limit of ORBGRAND and to focus on the performance loss due to reliability ranking, we consider the untruncated setting with $Q=2^n$.
	
	As shown in \cite[Sec.~II]{liu2022orbgrand}, with $Q=2^n$, ORBGRAND is equivalent to the following decoding rule:
	\begin{align}
		\label{decoding rule}
		\hat{m}=& \arg\min_{m=1,\ldots,\lceil e^{nR}\rceil}\mathsf{D}(\underline{x}(m),\underline{y}),\nonumber\\
		\mathsf{D}(\underline{x}(m),\underline{y})
		=&\frac{1}{n}\sum_{i=1}^{n}\frac{r_i}{n}
		\mathbf{1}\bigl(\operatorname{sgn}(l_i)\cdot x_i(m)<0\bigr),
	\end{align}
	where $\frac{r_i}{n} \mathbf{1}\bigl(\operatorname{sgn}(l_i)\cdot x_i(m)<0\bigr)$
	is the decoding metric associated with the $i$-th symbol position, borrowing the standard terminology in mismatched decoding \cite{lapidoth1998reliable,scarlett2020information}, and we call $\mathsf{D}(\underline{x}(m),\underline{y})$ the codeword decoding metric. Since each rank $r_i$ is determined by the global reliability ranking, these per-symbol decoding metrics are coupled across symbol positions.
	
	In the next
	section, we specialize the general mismatched RCU bound to this
	rank-based decoding rule and characterize the transmitted-codeword
	decoding metric required for the subsequent asymptotic analysis.
	
	\section{Random-Coding Union Bound and Transmitted-Codeword Decoding Metric Analysis}\label{sec:rcu}
	
	This section develops the two ingredients needed for the
	subsequent finite-blocklength analysis. We first specialize the
	general mismatched RCU bound to the rank-based ORBGRAND rule and then characterize the distributional behavior of the
	transmitted-codeword decoding metric.
	
	\subsection{Random-Coding Union Bound of ORBGRAND}
	
	Under the i.i.d.\ random-coding ensemble with $P_{\underline{\mathsf{X}}}=\prod_{i=1}^{n}P_{\mathsf{X}_i}$ and
	$P_{\mathsf{X}_i}(+1)=P_{\mathsf{X}_i}(-1)=1/2$, an upper bound on the ensemble-average decoding error probability of ORBGRAND is given below in Theorem \ref{ORB RCU}. The RCU bound was first introduced in
	\cite{polyanskiy2010channel}, with its extension to mismatched decoding established in \cite{martinez2011saddlepoint}
	and subsequently re-derived in \cite{scarlett2014saddlepoint}.
	Specializing this general mismatched RCU bound to the rank-based
	decoding metric of ORBGRAND gives Theorem~\ref{ORB RCU}.
	
	\begin{theorem}[ORB-RCU bound]\label{ORB RCU}
		For the random codebook generated i.i.d.\ according to $P_{\underline{\mathsf{X}}}$, the ensemble-average decoding error probability of ORBGRAND, denoted by $\epsilon(n,M)$, satisfies
		\begin{equation}
			\epsilon(n,M) \le \mathrm{RCU}_{\mathrm{ORB}}(n,M),
		\end{equation}
		where
		\begin{equation}\label{eq:rcu_n_M}
			\mathrm{RCU}_{\mathrm{ORB}}(n,M)
			\triangleq
			\mathbb{E}\left[
			\min\left\{
			1,
			(M-1)
			\Pr\left[
			\mathsf{D}(\underline{\hat{\mathsf{X}}},\underline{\mathsf{Y}})
			\le
			\mathsf{D}(\underline{\mathsf{X}},\underline{\mathsf{Y}})
			\big|
			\underline{\mathsf{X}},\underline{\mathsf{Y}}
			\right]
			\right\}
			\right],
		\end{equation}
		and $(\underline{\mathsf{X}},\underline{\mathsf{Y}},\underline{\hat{\mathsf{X}}})
		\sim
		P_{\underline{\mathsf{X}}}(\underline{x})
		P_{\underline{\mathsf{Y}}|\underline{\mathsf{X}}}(\underline{y}| \underline{x})
		P_{\underline{\mathsf{X}}}(\underline{\hat{x}})$.
	\end{theorem}
	
	\begin{IEEEproof}
		The proof is based on the general random-coding framework for
		the mismatched RCU bound
		\cite{martinez2011saddlepoint,scarlett2014saddlepoint}.  We provide the details needed to specialize this framework to ORBGRAND and to identify the pairwise probability analyzed below.
		
		For a fixed codebook $\bigl(\underline{x}(1),\ldots,\underline{x}(M)\bigr)$, the average decoding error probability of ORBGRAND satisfies, by the union bound,
		\begin{align}
			\epsilon(n,M)
			&\le \frac{1}{M}\sum_{m=1}^{M}
			\Pr\Big[
			\bigcup_{\substack{j=1, j\neq m}}^{M}
			\Bigl\{\mathsf{D}(\underline{x}(j),\underline{\mathsf{Y}})
			\le \mathsf{D}(\underline{x}(m),\underline{\mathsf{Y}})\Bigr\}
			\Big|
			\underline{\mathsf{X}}=\underline{x}(m)
			\Big].
		\end{align}
		Under the i.i.d. random-coding ensemble, the ensemble-average decoding error probability equals the conditional decoding error probability given that message $m=1$ is transmitted, i.e.,
		$\underline{\mathsf{X}}=\underline{\mathsf{X}}(1)$. Hence,
		\begin{align}
			\epsilon(n,M)
			&\le
			\Pr\Big[
			\bigcup_{j=2}^{M}
			\Bigl\{\mathsf{D}(\underline{\mathsf{X}}(j),\underline{\mathsf{Y}})
			\le \mathsf{D}(\underline{\mathsf{X}}(1),\underline{\mathsf{Y}})\Bigr\}
			\Big] \nonumber\\
			&=
			\mathbb{E}\left[
			\Pr\Big[
			\bigcup_{j=2}^{M}
			\Bigl\{\mathsf{D}(\underline{\mathsf{X}}(j),\underline{\mathsf{Y}})
			\le \mathsf{D}(\underline{\mathsf{X}}(1),\underline{\mathsf{Y}})\Bigr\}
			\Big|\underline{\mathsf{X}}(1),\underline{\mathsf{Y}} \Big]
			\right]\nonumber\\
			&\le
			\mathbb{E}\left[
			\min\left\{1, (M-1)\Pr\left[
			\mathsf{D}(\underline{\hat{\mathsf{X}}},\underline{\mathsf{Y}})
			\le \mathsf{D}(\underline{\mathsf{X}},\underline{\mathsf{Y}})
			\Big|\underline{\mathsf{X}},\underline{\mathsf{Y}}\right] \right\}
			\right],
		\end{align}
		where $(\underline{\mathsf{X}},\underline{\mathsf{Y}},\underline{\hat{\mathsf{X}}})\sim
		P_{\underline{\mathsf{X}}}(\underline{x})P_{\underline{\mathsf{Y}}|\underline{\mathsf{X}}}(\underline{y}|\underline{x})
		P_{\underline{\mathsf{X}}}(\underline{\hat{x}})$.	
	\end{IEEEproof}
	
	Based on Theorem \ref{ORB RCU}, we see that the resulting upper bound on the ensemble-average decoding error probability of ORBGRAND is fully determined by 
	$\Pr\left[
	\mathsf{D}(\underline{\hat{\mathsf{X}}},\underline{\mathsf{Y}})
	\le \mathsf{D}(\underline{\mathsf{X}},\underline{\mathsf{Y}})
	\Big|\underline{\mathsf{X}},\underline{\mathsf{Y}}\right]$. For $\mathsf{D}(\underline{\hat{\mathsf{X}}},\underline{\mathsf{Y}})$, since $\underline{\hat{\mathsf{X}}}$ and $\underline{\mathsf{Y}}$ are independent, we obtain:
	\begin{align}
		\mathsf{D}(\underline{\hat{\mathsf{X}}},\underline{\mathsf{Y}})=&\frac{1}{n^2}\sum_{i=1}^{n}\mathsf{R}_i\mathbf{1}\bigl(\operatorname{sgn}(\mathsf{L}_i)\cdot \hat{\mathsf{X}}_i<0\bigr),
	\end{align}
	in which $\mathbf{1}\bigl(\operatorname{sgn}(\mathsf{L}_i)\cdot\hat{\mathsf{X}}_i<0\bigr)$,
	$i=1,\ldots,n$, are i.i.d. Bernoulli random variables with parameter $1/2$. Therefore, we can equivalently write $\mathsf{D}(\underline{\hat{\mathsf{X}}},\underline{\mathsf{Y}})=\frac{1}{n^2} \sum_{i = 1}^n i \mathsf{B}_i$ where $\mathsf{B}_i, i = 1, \ldots, n$, are i.i.d. Bernoulli random variables with parameter $1/2$. Let $\zeta_n \triangleq \sum_{i=1}^{n} i\mathsf{B}_i = n^2\mathsf{D}(\underline{\hat{\mathsf{X}}},\underline{\mathsf{Y}})$. Since conditioned upon $(\underline{\mathsf{X}}, \underline{\mathsf{Y}})$ the randomness in $\mathsf{D}(\underline{\hat{\mathsf{X}}},\underline{\mathsf{Y}})$ is only induced by $\underline{\hat{\mathsf{X}}}$, we obtain
	\begin{align}
		\Pr\left[
		\mathsf{D}(\underline{\hat{\mathsf{X}}},\underline{\mathsf{Y}})
		\le
		\mathsf{D}(\underline{\mathsf{X}},\underline{\mathsf{Y}})
		\Big|
		\underline{\mathsf{X}},\underline{\mathsf{Y}}
		\right]
		&=
		\Pr\left[
		\zeta_n \le n^2\mathsf{D}(\underline{\mathsf{X}},\underline{\mathsf{Y}})
		\Big|
		\underline{\mathsf{X}},\underline{\mathsf{Y}}
		\right] =
		F_{\zeta_n}\left(n^2\mathsf{D}(\underline{\mathsf{X}},\underline{\mathsf{Y}})\right),
	\end{align}
	where $F_{\zeta_n}(\cdot)$ denotes the CDF of $\zeta_n$.
	
	In line with the tail-probability representation of the RCU bound
	considered in \cite{martinez2011saddlepoint} and subsequently
	used in, e.g., \cite{scarlett2014saddlepoint,scarlett2014mismatched}, we invoke the identity
	\begin{equation}\label{eq:min-to-unif}
		\mathbb{E}\left[\min\{1,\mathsf{A}\}\right]
		=\Pr\left[\mathsf{A}\ge \mathsf{U}\right],
	\end{equation}
	where $\mathsf{A}\ge 0$ and $\mathsf{U}\sim \mathrm{Unif}[0,1]$ is independent of $\mathsf{A}$. Applying \eqref{eq:min-to-unif} to the ORB-RCU
	bound gives
	\begin{align}
		\mathrm{RCU}_{\mathrm{ORB}}(n,M)\label{rcu n M}
		&=\mathbb{E}\left[\min\left\{1,(M-1)F_{\zeta_n}\bigl(n^2\mathsf{D}(\underline{\mathsf{X}},\underline{\mathsf{Y}})\bigr)\right\}\right] \nonumber\\
		&=\Pr\left[(M-1)F_{\zeta_n}\bigl(n^2\mathsf{D}(\underline{\mathsf{X}},\underline{\mathsf{Y}})\bigr)\ge \mathsf{U}\right] \nonumber\\
		&=\Pr\left[\ln(M-1)+\ln F_{\zeta_n}\bigl(n^2\mathsf{D}(\underline{\mathsf{X}},\underline{\mathsf{Y}})\bigr)-\ln \mathsf{U}\ge 0\right].
	\end{align}
	
	The representation in \eqref{rcu n M} shows that the
	finite-blocklength analysis of $\mathrm{RCU}_{\mathrm{ORB}}(n,M)$ requires characterizing both the transmitted-codeword decoding metric and the distribution function
	$F_{\zeta_n}(\cdot)$ associated with the competing codeword. We analyze the former in the next
	subsection, while the latter is treated in
	Section~\ref{sec:strong}.

	\subsection{Transmitted-Codeword Decoding Metric Analysis}
	
	In this subsection, we characterize the asymptotic behavior of the transmitted-codeword decoding metric $\mathsf{D}(\underline{\mathsf{X}},\underline{\mathsf{Y}})$. Since
	$\mathsf R_i/n$ is the empirical quantile of $\Lambda_i$ among
	the $n$ reliability magnitudes, it is expected to be close to
	the corresponding population quantile $\Psi(\Lambda_i)$ for
	large $n$. This motivates the following heuristic approximation
	\begin{align}
		\mathsf D(\underline{\mathsf X},\underline{\mathsf Y})
		=
		\frac{1}{n}
		\sum_{i=1}^{n}
		\frac{\mathsf R_i}{n}\mathsf E_i
		\approx
		\frac{1}{n}
		\sum_{i=1}^{n}
		\Psi(\Lambda_i)\mathsf E_i.
		\label{eq:intuition-transmitted-metric}
	\end{align}
	The quantity on the right approaches, by the law of large
	numbers, to $\mathbb E[\Psi(\Lambda)\mathsf E]=\mu$.
	This observation motivates centering the subsequent
	finite-blocklength analysis at $\mu$. It does not by itself
	characterize the fluctuations of the transmitted-codeword decoding metric, since each rank
	$\mathsf R_i$ depends jointly on all the reliability
	magnitudes.
	
	We begin by showing that $\mu$ lies strictly inside the interval $(0,1/4)$. We then represent the transmitted-codeword decoding metric, up to a deterministic $O(n^{-1})$ term, as a nondegenerate second-order $U$-statistic and apply a Berry-Esseen theorem for $U$-statistics to obtain a uniform Gaussian approximation.

	The following lemma shows that $\mu$ lies strictly inside the interval $(0,1/4)$.
	
	\begin{lemma}\label{conv point}
		Under Assumption~\ref{assump:regularity}, the quantity $\mu=\mathbb{E}[\Psi(\Lambda)\mathsf E]$ satisfies $0<\mu<\frac{1}{4}.$
	\end{lemma}
	\begin{IEEEproof}
		See Appendix \ref{proof conv point}.
	\end{IEEEproof}

	From Lemma~\ref{conv point}, we therefore fix constants $\omega_1$ and $\omega_2$ such that
	\begin{align}
		0<\omega_1<\mu<\omega_2<\frac14,
		\qquad
		\mathcal D\triangleq[\omega_1,\omega_2].
		\label{eq:d-compact-set}
	\end{align}

	Furthermore, we establish an asymptotic Gaussian characterization of the transmitted-codeword decoding metric $\mathsf{D}(\underline{\mathsf{X}},\underline{\mathsf{Y}})$ around $\mu$. The following theorem characterizes the fluctuations of the
	transmitted-codeword decoding metric at the $n^{-1/2}$ scale and
	provides a uniform Gaussian approximation for its distribution. A central difficulty is that $\mathsf{D}(\underline{\mathsf{X}},\underline{\mathsf{Y}})$ depends on the global reliability ordering and therefore involves cross-symbol coupling, so standard second-order analyses for sums of independent single-letter decoding metrics do not apply directly.
	To overcome this difficulty, we represent $\mathsf{D}(\underline{\mathsf{X}},\underline{\mathsf{Y}})$ as a second-order $U$-statistic up to a deterministic $O(n^{-1})$ term and analyze its first-order Hoeffding projection \cite{hoeffding1992class}. This projection identifies the leading fluctuation and yields a single-letter expression for the asymptotic variance. A Berry-Esseen theorem for nondegenerate $U$-statistics then provides a uniform Gaussian approximation for the transmitted-codeword decoding metric.

	To state the Gaussian approximation below, let
	$(\mathsf E',\Lambda')$ be an independent copy of
	$(\mathsf E,\Lambda)$, and define
	\begin{align}
		a(x)
		&\triangleq
		\Pr[\mathsf E'=1,\Lambda'\ge x],
		\label{eq:a-function}\\
		\sigma^2
		&\triangleq
		\operatorname{Var}\!\left(
		\mathsf E\Psi(\Lambda)+a(\Lambda)
		\right).
		\label{eq:metric-variance-parameter}
	\end{align}
	
	\begin{theorem}[Berry-Esseen approximation for the
		transmitted-codeword decoding metric]
		\label{thm:gauss of full D}
		Under Assumption~\ref{assump:regularity}, the variance
		parameter $\sigma^2$ defined in
		\eqref{eq:metric-variance-parameter} is strictly positive.
		Moreover, with $\sigma>0$ denoting the corresponding standard deviation, there
		exists a constant $0<C<\infty$, independent of $n$, such
		that, for all $n\ge2$,
		\begin{align}
			\sup_{t\in\mathbb R}
			\left|
			\Pr\left[
			\frac{
				\sqrt n\bigl(
				\mathsf D(\underline{\mathsf X},
				\underline{\mathsf Y})-\mu
				\bigr)
			}{\sigma}
			\ge t
			\right]
			-
			Q(t)
			\right|
			\le
			\frac{C}{\sqrt n},
		\end{align}
		where $Q(\cdot)$ is the Gaussian tail function.
	\end{theorem}
	\begin{IEEEproof}
		See Appendix \ref{proof gauss full D}.
	\end{IEEEproof}

	The preceding theorem provides the Gaussian approximation
	required for the transmitted-codeword term in
	\eqref{rcu n M}. It remains to characterize the term $F_{\zeta_n}(\cdot)$ associated with the competing codeword; the next section develops a uniform
	strong large-deviation expansion for this quantity over a
	neighborhood of $\mu$.

	\section{Strong Large-Deviation Analysis for Competing-Codeword Decoding Metric}\label{sec:strong}
	
	In this section, we study the competing-codeword decoding metric
	$\mathsf{D}(\underline{\hat{\mathsf{X}}},\underline{\mathsf{Y}})$.
	Recall that $\zeta_n = \sum_{i=1}^{n} i\mathsf{B}_i
	= n^2\mathsf{D}(\underline{\hat{\mathsf{X}}},\underline{\mathsf{Y}})$, so the distribution of the competing-codeword decoding metric is determined by the left-tail probability $F_{\zeta_n}(n^2d)$ over the range of $d$ relevant to \eqref{rcu n M}. To support the subsequent analysis, we establish a uniform strong large-deviation expansion for $F_{\zeta_n}(n^2d)$ over the relevant range of $d$. Strong large-deviation expansions can trace back to the classical Bahadur-Rao result \cite{bahadur1960deviations} and its extensions to more general sequences and lattice-valued
	settings \cite{chaganty1993strong,joutard2013strong}; they have
	also been used to derive asymptotic bounds on the
	maximum code log-volume \cite{moulin2017log}. These existing
	results do not directly yield the expansion required here,
	because $\zeta_n=\sum_{i=1}^{n}i\mathsf B_i$ is a weighted lattice sum of independent but non-identically distributed summands, with
	weights varying with the blocklength. Moreover, the subsequent
	analysis requires an explicit subexponential prefactor
	and a remainder that hold uniformly for $d$ in a compact
	interval. We therefore exploit the specific weighted-Bernoulli structure and derive the required expansion using exponential tilting, uniform finite-blocklength saddlepoint expansions, and a lattice local limit argument.

	\subsection{Setup and Cumulant Generating Functions}

	For the large-deviation analysis below, define the normalized
	weighted Bernoulli sum
	\begin{align}
		\mathsf{H}_n
		\triangleq
		\frac{1}{n^2}
		\sum_{i=1}^{n}i\mathsf{B}_i.
	\end{align}
	The moment generating function (MGF) and the normalized cumulant
	generating function (CGF) of $n\mathsf{H}_n$ are defined,
	respectively, as
	\begin{align}
		\Upsilon_n(\theta)
		\triangleq
		\mathbb{E}\left[
		e^{\theta n\mathsf{H}_n}
		\right],
		\qquad
		K_n(\theta)
		\triangleq
		\frac{1}{n}\ln\Upsilon_n(\theta).
	\end{align}

	Since $\{\mathsf{B}_i\}_{i=1}^n$ are i.i.d. Bernoulli random variables with parameter $1/2$, we have
	\begin{align}
		\Upsilon_n(\theta)
		&=
		\prod_{i=1}^{n}
		\frac{1+e^{\theta i/n}}{2},\\
		K_n(\theta)
		&=
		\frac{1}{n}
		\sum_{i=1}^{n}
		\ln\frac{1+e^{\theta i/n}}{2}.
		\label{eq:Kn-explicit}
	\end{align}
	Consequently, for every $\theta\in\mathbb{R}$, a standard
	Riemann-sum argument gives
	\begin{align}\label{eq:K}
		K_n(\theta)
		\longrightarrow
		K(\theta)
		\triangleq
		\int_{0}^{1}
		\ln\frac{1+e^{\theta x}}{2}
		\,\mathrm{d}x.
	\end{align}
	The associated Legendre transform is defined as
	\begin{align}\label{eq:legendre}
		I(d)
		\triangleq
		\sup_{\theta\in\mathbb{R}}
		\bigl\{\theta d-K(\theta)\bigr\}.
	\end{align}

	Recalling $\zeta_n=\sum_{i=1}^{n}i\mathsf{B}_i=n^2\mathsf{H}_n$, we have
	\begin{align}
		F_{\zeta_n}(n^2 d)=\Pr[\zeta_n\le n^2 d]=\Pr[\mathsf{H}_n\le d].
	\end{align}
	
	The expressions above characterize the limiting CGF and the associated rate function. Our next goal is to derive a uniform strong large-deviation expansion for $\Pr[\mathsf{H}_n\le d]$ when $d$ ranges over a compact subset of $(0,1/4)$. To this end, we directly analyze the left-tail probability by exponential tilting at a finite-blocklength saddlepoint and by establishing a uniform lattice local central limit theorem under the tilted measure.

	\subsection{Finite-Blocklength Saddlepoint and Exponential Tilting}

	We first characterize the limiting saddlepoint associated with the
	left-tail event. Differentiating \eqref{eq:K} under the integral sign
	gives
	\begin{align}
		K'(\theta)
		&=
		\int_0^1
		\frac{xe^{\theta x}}
		{1+e^{\theta x}}
		\,\mathrm{d}x,
		\label{eq:K-first-derivative}\\
		K''(\theta)
		&=
		\int_0^1
		\frac{x^2e^{\theta x}}
		{\left(1+e^{\theta x}\right)^2}
		\,\mathrm{d}x
		>0.
		\label{eq:K-second-derivative}
	\end{align}
	Hence, $K'(\cdot)$ is continuous and strictly increasing. Its behavior over the negative half-line is characterized by
	\begin{align}
		\lim_{\theta\to-\infty}K'(\theta)
		=0,
		\qquad
		K'(0)=\frac14.
		\label{eq:K-prime-range}
	\end{align}
	It follows that, for every $d\in(0,1/4)$, there exists a unique
	$\theta_d<0$ satisfying $K'(\theta_d)=d$. Since $\theta\mapsto\theta d-K(\theta)$ is strictly concave,
	$\theta_d$ is its unique maximizer. The rate function in
	\eqref{eq:legendre} therefore admits the representation
	\begin{align}
		I(d)
		=
		\theta_dd-K(\theta_d),
		\qquad
		d\in(0,1/4).
		\label{eq:rate-function-saddlepoint}
	\end{align}
	
	We next establish a strong large-deviation expansion uniformly
	over the compact interval $\mathcal D$ defined in
	\eqref{eq:d-compact-set}. Since the unique solution $\theta_d$ to $K'(\theta_d)=d$ depends continuously on $d$, the set $\{\theta_d:d\in\mathcal D\}$ is a compact subset of $(-\infty,0)$. This compactness, together with the continuity and strict positivity of $K''(\cdot)$, ensures that $\inf_{d\in\mathcal D}K''(\theta_d)>0$.
		
	We account for the lattice structure of $\zeta_n$ when defining
	the finite-blocklength saddlepoint. Since $\zeta_n$ is integer-valued, introduce the lattice-adjusted threshold
	\begin{align}
		\kappa_n(d)
		\triangleq
		\left\lfloor n^2d\right\rfloor,
		\qquad
		d_n
		\triangleq
		\frac{\kappa_n(d)}{n^2}.
		\label{eq:qn-dn-definition}
	\end{align}
	These definitions yield the exact event identity
	\begin{align}
		\{\mathsf H_n\le d\}
		=
		\{\zeta_n\le \kappa_n(d)\}
		=
		\{\mathsf H_n\le d_n\},
		\label{eq:lattice-event-equivalence}
	\end{align}
	while the corresponding rounding error satisfies
	\begin{align}
		0
		\le
		d-d_n
		<
		\frac{1}{n^2}.
		\label{eq:dn-rounding}
	\end{align}
	
	For later use, define the finite-blocklength rate function and
	the Euler-Maclaurin boundary correction as
	\begin{align}
		I_n(u)
		&\triangleq
		\sup_{\theta\in\mathbb R}
		\bigl\{
		\theta u-K_n(\theta)
		\bigr\},
		\nonumber\\
		G(\theta)
		&\triangleq
		\frac12
		\ln\frac{1+e^\theta}{2}.
		\label{eq:finite-rate-function-and-correction}
	\end{align}
	The following lemma collects the uniform saddlepoint
	expansions needed for the exponential-tilting argument.
	\begin{lemma}\label{lem:finite-saddlepoint}
		For all sufficiently large $n$ and every
		$d\in\mathcal D$, the equation
		\begin{align}
			K_n'(\theta_{n,d})
			=
			d_n
			\label{eq:finite-blocklength-saddlepoint}
		\end{align}
		has a unique solution $\theta_{n,d}$, and these solutions
		remain in a fixed compact subset of $(-\infty,0)$. In addition, $\theta_{n,d}$ is the unique maximizer in the definition of $I_n(d_n)$, and hence
		\begin{align}
			I_n(d_n)
			=
			\theta_{n,d}d_n-K_n(\theta_{n,d}).
			\label{eq:finite-blocklength-rate-function}
		\end{align}
		Moreover, uniformly over $d\in\mathcal D$,
		\begin{align}
			\theta_{n,d}
			&=
			\theta_d+O(n^{-1}),
			\label{eq:theta-nd-expansion}\\
			K_n''(\theta_{n,d})
			&=
			K''(\theta_d)+O(n^{-1}),
			\label{eq:Kn-second-expansion}\\
			nI_n(d_n)
			&=
			nI(d)-G(\theta_d)+O(n^{-1}).
			\label{eq:finite-rate-function-expansion}
		\end{align}
	\end{lemma}

	\begin{IEEEproof}
		See Appendix \ref{proof saddlepoint expansion}.
	\end{IEEEproof}
	
	Based on \eqref{eq:qn-dn-definition} and \eqref{eq:finite-blocklength-saddlepoint}, we have
	\begin{align}
		n^2K_n'(\theta_{n,d})
		=
		n^2d_n
		=
		\kappa_n(d).
		\label{eq:exact-lattice-centering}
	\end{align}
	This exact centering motivates an exponential change of measure under which the lattice threshold becomes the mean of $\zeta_n$. This change-of-measure construction is motivated by the
	exponential-tilting method underlying Bahadur-Rao type strong
	large-deviation analyses \cite{bahadur1960deviations,chaganty1993strong,joutard2013strong}.
	
	Let $\Pr$ denote the original probability measure governing
	$\zeta_n$. For each $d\in\mathcal D$, define the exponentially tilted probability measure $\Pr\nolimits_{n,d}$ through
	\begin{align}
		\frac{\mathrm{d}\Pr\nolimits_{n,d}}{\mathrm{d}\Pr}
		\triangleq
		\exp\left(
		\frac{\theta_{n,d}}{n}\zeta_n
		-
		nK_n(\theta_{n,d})
		\right).
		\label{eq:tilted-measure}
	\end{align}
	The normalization follows directly from the definition of $K_n(\cdot)$, since
	\begin{align}
		\mathbb E\left[
		\exp\left(
		\frac{\theta_{n,d}}{n}\zeta_n
		-
		nK_n(\theta_{n,d})
		\right)
		\right]
		=
		1.
	\end{align}
	
	Under $\Pr\nolimits_{n,d}$, the random variables
	$\{\mathsf B_i\}_{i=1}^{n}$ remain independent and satisfy
	\begin{align}
		\Pr\nolimits_{n,d}[\mathsf B_i=1]
		=
		\frac{e^{\theta_{n,d}i/n}}
		{1+e^{\theta_{n,d}i/n}},
		\qquad
		i\in\{1,\ldots,n\}.
		\label{eq:tilted-Bernoulli-parameter}
	\end{align}
	The first two derivatives of $K_n(\cdot)$ therefore give the tilted mean and variance of $\zeta_n$ as
	\begin{align}
		\mathbb E_{n,d}[\zeta_n]
		&=
		n^2K_n'(\theta_{n,d})
		=
		\kappa_n(d),
		\label{eq:tilted-mean}\\
		\sigma_{n,d}^2
		&\triangleq
		\operatorname{Var}_{n,d}(\zeta_n)
		=
		n^3K_n''(\theta_{n,d}).
		\label{eq:tilted-variance}
	\end{align}
	Here, $\sigma_{n,d}>0$ denotes the corresponding
	standard deviation. In view of \eqref{eq:Kn-second-expansion}, the tilted variance satisfies the uniform expansion
	\begin{align}
		\sigma_{n,d}^2
		=
		n^3K''(\theta_d)
		\bigl(1+O(n^{-1})\bigr),
		\qquad
		d\in\mathcal D.
		\label{eq:tilted-variance-expansion}
	\end{align}
	In particular, there exist constants $0<C_0<C_1<\infty$ such that, for all sufficiently large $n$,
	\begin{align}
		C_0n^3
		\le
		\sigma_{n,d}^2
		\le
		C_1n^3,
		\qquad
		d\in\mathcal D.
		\label{eq:tilted-variance-order}
	\end{align}
	
	We next use the tilted measure to separate the exponential decay from the local lattice contribution. By \eqref{eq:lattice-event-equivalence} and \eqref{eq:tilted-measure}, the left-tail probability can be written exactly as
	\begin{align}
		\Pr[\mathsf H_n\le d]
		&=
		\mathbb E_{n,d}
		\left[
		\exp\left(
		-\frac{\theta_{n,d}}{n}\zeta_n
		+
		nK_n(\theta_{n,d})
		\right)
		\mathbf{1}\{\zeta_n\le \kappa_n(d)\}
		\right].
		\label{eq:change-of-measure-expectation}
	\end{align}
	For each realization satisfying $\zeta_n\le \kappa_n(d)$, introduce the nonnegative lattice displacement $\ell\triangleq\kappa_n(d)-\zeta_n$. Using $\kappa_n(d)=n^2d_n$ and \eqref{eq:finite-blocklength-rate-function}, we obtain
	\begin{align}
		-\frac{\theta_{n,d}}{n}\bigl(\kappa_n(d)-\ell\bigr)
		+
		nK_n(\theta_{n,d})
		=
		-nI_n(d_n)
		+
		\frac{\theta_{n,d}}{n}\ell.
	\end{align}
	Consequently, the left-tail probability admits the  decomposition
	\begin{align}
		\Pr[\mathsf H_n\le d]
		=
		e^{-nI_n(d_n)}T_{n,d},
		\label{eq:exact-tilting-decomposition}
	\end{align}
	where the one-sided lattice term is defined by
	\begin{align}
		T_{n,d}
		\triangleq
		\sum_{s=0}^{\kappa_n(d)}
		e^{\theta_{n,d}s/n}
		\Pr\nolimits_{n,d}
		\left[
		\zeta_n=\kappa_n(d)-s
		\right].
		\label{eq:one-sided-lattice-term}
	\end{align}
	
	The decomposition in \eqref{eq:exact-tilting-decomposition} isolates
	the finite-blocklength large-deviation exponent factor
	$e^{-nI_n(d_n)}$ from the one-sided lattice contribution $T_{n,d}$.
	The latter will be evaluated in the next subsection through a uniform lattice local central limit theorem.

	\subsection{Uniform Lattice Approximation and Strong Large Deviations}
	
	We now evaluate the one-sided lattice contribution $T_{n,d}$ in
	\eqref{eq:one-sided-lattice-term}. Under the tilted measure
	$\Pr\nolimits_{n,d}$, the threshold $\kappa_n(d)$ coincides exactly with the
	mean of $\zeta_n$, while its variance is of order $n^3$. We first
	establish a lattice local central limit approximation that holds
	uniformly over the family of tilted distributions generated by
	$d\in\mathcal D$.
	
	The required approximation is summarized in the following
	lemma. Its proof follows the standard Fourier-inversion
	approach to lattice local central limit theorems for triangular
	arrays, based on a uniform Lyapunov-type moment bound and
	characteristic function decay away from the origin; see, e.g.,
	\cite{petrov2012sums}.
	
	\begin{lemma}\label{lem:uniform-lattice-lclt}
		The lattice probabilities of $\zeta_n$ under the tilted measures
		$\{\Pr\nolimits_{n,d}:d\in\mathcal D\}$ admit the following Gaussian
		approximation uniformly over $d\in\mathcal D$ and the lattice
		points $k\in\mathbb Z$:
		\begin{align}
			\lim_{n\to\infty}
			\sup_{d\in\mathcal D}
			\sup_{k\in\mathbb Z}
			\left|
			\sigma_{n,d}
			\Pr\nolimits_{n,d}[\zeta_n=k]
			-
			\phi\left(
			\frac{k-\kappa_n(d)}
			{\sigma_{n,d}}
			\right)
			\right|
			=
			0,
			\label{eq:uniform-lattice-lclt}
		\end{align}
		where $\phi(x)\triangleq\frac{1}{\sqrt{2\pi}}e^{-x^2/2}$ denotes the standard Gaussian probability density function.
	\end{lemma}
	\begin{IEEEproof}
		See Appendix~\ref{app:uniform-lattice-lclt}.
	\end{IEEEproof}
	
	With the uniform local Gaussian approximation established, we now
	evaluate the one-sided lattice contribution in
	\eqref{eq:one-sided-lattice-term}. The resulting asymptotic expansion is stated in the following lemma.
	
	\begin{lemma}\label{lem:one-sided-lattice-contribution}
		The one-sided lattice contribution admits the uniform asymptotic expansion
		\begin{align}
			T_{n,d}
			=
			\frac{
				1
			}{
				|\theta_d|
				\sqrt{2\pi nK''(\theta_d)}
			}\bigl(1+\rho_n(d)\bigr),
			\qquad
			d\in\mathcal D,
			\label{eq:one-sided-lattice-contribution-expansion}
		\end{align}
		where $\sup_{d\in\mathcal D}|\rho_n(d)|\to0$ as $n\to\infty$.
	\end{lemma}
	\begin{IEEEproof}
		See Appendix~\ref{app:one-sided-lattice-contribution}.
	\end{IEEEproof}
	
	Combining the exact change-of-measure decomposition in
	\eqref{eq:exact-tilting-decomposition}, the finite-blocklength
	rate-function expansion in \eqref{eq:finite-rate-function-expansion}, and the uniform asymptotics of the one-sided lattice contribution in
	Lemma~\ref{lem:one-sided-lattice-contribution}, we obtain the desired uniform strong large-deviation expansion.
	
	\begin{theorem}[Uniform strong large-deviation expansion]
		\label{thm:uniform-strong-large-deviation}
		For the compact interval
		$\mathcal D=[\omega_1,\omega_2]\subset(0,1/4)$, define
		\begin{align}
			A(d)
			\triangleq
			\sqrt{
				\frac{
					1+e^{\theta_d}
				}{
					4\pi K''(\theta_d)\theta_d^2
				}
			},
			\qquad
			d\in\mathcal D.
			\label{eq:strong-large-deviation-prefactor}
		\end{align}
		Then there exists a sequence of remainder functions
		$\varrho_n(\cdot):\mathcal D\to\mathbb R$ satisfying
		$\sup_{d\in\mathcal D}|\varrho_n(d)|\to0$ as $n\to\infty$, such
		that
		\begin{align}
			\Pr[\mathsf H_n\le d]
			=
			\frac{A(d)}{\sqrt n}
			e^{-nI(d)}
			\bigl(1+\varrho_n(d)\bigr),
			\qquad
			d\in\mathcal D.
			\label{eq:uniform-strong-large-deviation}
		\end{align}
	\end{theorem}
	\begin{IEEEproof}
		The finite-blocklength rate-function expansion in
		\eqref{eq:finite-rate-function-expansion} implies, uniformly over
		$d\in\mathcal D$, that
		\begin{align}
			e^{-nI_n(d_n)}
			=
			e^{-nI(d)}
			e^{G(\theta_d)}
			\bigl(1+O(n^{-1})\bigr).
			\label{eq:finite-exponential-factor-expansion}
		\end{align}
		The definition of $G(\cdot)$ in \eqref{eq:finite-rate-function-and-correction} further gives
		\begin{align}
			e^{G(\theta_d)}
			=
			\sqrt{
				\frac{1+e^{\theta_d}}{2}
			}.
			\label{eq:boundary-correction-exponential}
		\end{align}
		Combining these relations with the 
		decomposition in \eqref{eq:exact-tilting-decomposition} and the
		uniform expansion in
		Lemma~\ref{lem:one-sided-lattice-contribution}, we obtain
		\begin{align}
			\Pr[\mathsf H_n\le d]
			=
			\frac{e^{-nI(d)}}{\sqrt n}
			\sqrt{
				\frac{
					1+e^{\theta_d}
				}{
					4\pi K''(\theta_d)\theta_d^2
				}
			}
			\bigl(1+\varrho_n(d)\bigr),
		\end{align}
		where the combined remainder satisfies
		$\sup_{d\in\mathcal D}|\varrho_n(d)|\to0$ as $n\to\infty$. This completes the proof.
	\end{IEEEproof}
	
	The uniformity of the expansion over $\mathcal D$ will be
	essential in the next section, where it is combined with the
	Berry-Esseen approximation for the transmitted-codeword decoding metric to derive a Gaussian approximation of $\mathrm{RCU}_{\mathrm{ORB}}(n,M)$. 
	
	\section{Achievable Rate Expansion for ORBGRAND}
	\label{sec:second_third_order}
	
	In this section, we characterize the finite-blocklength
	achievable rate of ORBGRAND.
	We first employ a threshold relaxation of $\mathrm{RCU}_{\mathrm{ORB}}(n,M)$
	to obtain a concise second-order achievable rate result. Threshold-based random-coding arguments are discussed in
	\cite{strassen1962asymptotische,polyanskiy2010channel}, and a
	closely related threshold relaxation for mismatched decoding
	appears in \cite{scarlett2013cost}. We next carry out a direct analysis of $\mathrm{RCU}_{\mathrm{ORB}}(n,M)$. This refined analysis retains the logarithmic correction arising from the subexponential prefactor in the strong large-deviation expansion and yields the third-order achievable rate expansion.
	
	For $\epsilon\in(0,1)$, let $M_{\mathrm{ORB}}^\star(n,\epsilon)$ denote the largest integer $M$ for which there exists an $(n,M)$ code whose average error probability under ORBGRAND does not exceed $\epsilon$. The corresponding achievable rate is defined as
	\begin{align}
		R_{\mathrm{ORB}}^{\star}(n,\epsilon)
		\triangleq
		\frac{1}{n}
		\ln M_{\mathrm{ORB}}^{\star}(n,\epsilon).
		\label{eq:orb-optimal-achievable-rate}
	\end{align}
	
	We further define $V_{\mathrm{ORB}}\triangleq\theta_{\mu}^{2}\sigma^{2}$, where $\theta_{\mu}<0$ is the unique solution to
	$K'(\theta_{\mu})=\mu$, and $\sigma^{2}$ is the
	variance defined in \eqref{eq:metric-variance-parameter}.
	
	\subsection{A Threshold-Based Second-Order Achievable Rate Expansion}
	\label{subsec:threshold-second-order}
	
	We first derive a second-order achievable rate result through
	a threshold relaxation of $\mathrm{RCU}_{\mathrm{ORB}}(n,M)$. This approach separates the error event associated with the transmitted
	codeword from that associated with competing codewords,
	leading to a concise derivation of the first- and second-order
	terms.
	
	For each $n$ and any deterministic threshold
	$b_n\in\mathbb{R}$, the ORB-RCU bound \eqref{eq:rcu_n_M} satisfies
	\begin{align}
		\operatorname{RCU}_{\mathrm{ORB}}(n,M)
		\le{}&
		\Pr\left[
		\mathsf{D}\left(
		\underline{\mathsf X},
		\underline{\mathsf Y}
		\right)
		\ge b_n
		\right]
		+
		(M-1)
		\Pr\left[
		\mathsf{D}(
		\underline{\hat{\mathsf X}},
		\underline{\mathsf Y})
		<b_n
		\right].
		\label{eq:threshold-based-orb-rcu}
	\end{align}
	This threshold decomposition yields the following second-order achievable rate expansion.
	
	\begin{theorem}[Threshold-based second-order achievable rate expansion]
		\label{thm:threshold-second-order-rate}
		For every fixed $\epsilon\in(0,1)$, the ORBGRAND
		achievable rate satisfies
		\begin{align}
			R_{\mathrm{ORB}}^{\star}(n,\epsilon)
			\ge
			I(\mu)
			-
			\sqrt{\frac{V_{\mathrm{ORB}}}{n}}\,
			Q^{-1}(\epsilon)
			+
			O\left(\frac{1}{n}\right)
			\label{eq:threshold-second-order-achievable-rate}
		\end{align}
		for all sufficiently large $n$.
	\end{theorem}
	
	\begin{IEEEproof}
		Let $0<C_0<\infty$ be the constant in
		Theorem~\ref{thm:gauss of full D}, and fix any $C_0<C_1<\infty$.
		For all sufficiently large $n$, we have
		$\epsilon-C_1/\sqrt n\in(0,1)$. Choose
		\begin{align}
			b_n
			=
			\mu
			+
			\frac{\sigma}{\sqrt n}
			Q^{-1}\left(
			\epsilon-\frac{C_1}{\sqrt n}
			\right).
			\label{eq:threshold-choice}
		\end{align}
		The Berry-Esseen bound in
		Theorem~\ref{thm:gauss of full D} yields
		\begin{align}
			\Pr\left[
			\mathsf{D}\left(
			\underline{\mathsf X},
			\underline{\mathsf Y}
			\right)
			\ge b_n
			\right]
			\le
			\epsilon
			-
			\frac{C_1-C_0}{\sqrt n}.
			\label{eq:threshold-transmitted-term}
		\end{align}
		
		Moreover, $b_n\to\mu$, and hence
		$b_n\in\mathcal D$ for all sufficiently large $n$.
		Using the uniform strong large-deviation expansion in
		Theorem~\ref{thm:uniform-strong-large-deviation}, we obtain
		\begin{align}
			&\Pr\left[
			\mathsf{D}(
			\underline{\hat{\mathsf X}},
			\underline{\mathsf Y})
			<b_n
			\right]
			\le
			F_{\zeta_n}(n^2 b_n)
			=
			\frac{A(b_n)}{\sqrt n}
			e^{-nI(b_n)}
			\left(
			1+\varrho_n(b_n)
			\right).
			\label{eq:threshold-competing-sld}
		\end{align}
		Since $A(\cdot)$ is bounded on $\mathcal D$ and $\sup_{d\in\mathcal D}|\varrho_n(d)|\longrightarrow 0$, there exists a constant $0<C<\infty$ such that
		\begin{align}
			\Pr\left[
			\mathsf{D}(
			\underline{\hat{\mathsf X}},
			\underline{\mathsf Y})
			<b_n
			\right]
			\le
			\frac{C}{\sqrt n}
			e^{-nI(b_n)}
			\label{eq:threshold-competing-upper-bound}
		\end{align}
		for all sufficiently large $n$.
		
		We choose the codebook size $M_n$ as
		\begin{align}
			M_n
			=
			1+
			\left\lfloor
			\frac{C_1-C_0}{C}
			e^{nI(b_n)}
			\right\rfloor.
			\label{eq:threshold-codebook-choice}
		\end{align}
		It follows from
		\eqref{eq:threshold-competing-upper-bound} that
		\begin{align}
			&(M_n-1)
			\Pr\left[
			\mathsf{D}(
			\underline{\hat{\mathsf X}},
			\underline{\mathsf Y})
			<b_n
			\right]
			\le
			\frac{C_1-C_0}{\sqrt n}.
		\end{align}
		Combining this inequality with
		\eqref{eq:threshold-transmitted-term} and
		\eqref{eq:threshold-based-orb-rcu} gives $\operatorname{RCU}_{\mathrm{ORB}}(n,M_n)\le\epsilon$.
		Theorem~\ref{ORB RCU} therefore ensures that
		\begin{align}
			M_{\mathrm{ORB}}^{\star}(n,\epsilon)
			\ge M_n.
			\label{eq:threshold-Mstar-lower-bound}
		\end{align}
		
		The inequality $1+\lfloor x\rfloor\ge x$ for $x>0$
		further implies
		\begin{align}
			\ln M_n
			\ge
			nI(b_n)
			+
			\ln\frac{C_1-C_0}{C}.
			\label{eq:threshold-logM-lower-bound}
		\end{align}
		Since $Q^{-1}(\cdot)$ is continuously differentiable in a
		neighborhood of the fixed value $\epsilon$,
		\begin{align}
			b_n-\mu
			=
			\frac{\sigma}{\sqrt n}
			Q^{-1}(\epsilon)
			+
			O\left(\frac{1}{n}\right).
			\label{eq:threshold-beta-expansion}
		\end{align}
		Using $I'(\mu)=\theta_\mu$ and expanding $I(\cdot)$ around
		$\mu$, we obtain
		\begin{align}
			nI(b_n)
			&=
			nI(\mu)
			+
			\sqrt n\,\theta_\mu\sigma
			Q^{-1}(\epsilon)
			+
			O(1)
			\nonumber\\
			&=
			nI(\mu)
			-
			\sqrt{nV_{\mathrm{ORB}}}\,
			Q^{-1}(\epsilon)
			+
			O(1),
			\label{eq:threshold-rate-function-expansion}
		\end{align}
		where the second equality follows from
		$\theta_\mu<0$ and
		$V_{\mathrm{ORB}}=\theta_\mu^2\sigma^2$.
		
		Combining
		\eqref{eq:threshold-Mstar-lower-bound},
		\eqref{eq:threshold-logM-lower-bound}, and
		\eqref{eq:threshold-rate-function-expansion} yields
		\begin{align}
			\ln M_{\mathrm{ORB}}^{\star}(n,\epsilon)
			\ge
			nI(\mu)
			-
			\sqrt{nV_{\mathrm{ORB}}}\,
			Q^{-1}(\epsilon)
			+
			O(1).
		\end{align}
		Dividing both sides by $n$ proves
		\eqref{eq:threshold-second-order-achievable-rate}.
	\end{IEEEproof}
	
	The remainder $O(n^{-1})$ in
	Theorem~\ref{thm:threshold-second-order-rate} is a
	consequence of the threshold construction used in its proof.
	Specifically, the threshold relaxation separates the
	transmitted- and competing-codeword contributions and divides
	the available error probability budget between them. Although this decomposition is sufficient to establish the
	first- and second-order terms, it treats the transmitted- and
	competing-codeword contributions separately and therefore does
	not preserve their exact interaction in the ORB-RCU bound \eqref{eq:rcu_n_M}. As a result, obtaining a more refined expansion
	from the threshold relaxation bound becomes difficult. 
	
	Motivated by the tail-probability representation and direct
	saddlepoint analysis of the RCU bound developed in
	\cite{martinez2011saddlepoint} (see also
	\cite{scarlett2014saddlepoint,font2018saddlepoint,
		scarlett2014mismatched}), we next analyze the
	ORB-RCU bound \eqref{eq:rcu_n_M} without introducing an intermediate
	threshold relaxation. This direct treatment preserves the
	interaction between the transmitted- and competing-codeword
	terms, which is lost under the threshold relaxation. Although the underlying prefactor mechanism is conceptually similar to that arising for sums of independent single-letter decoding metrics, the existing analyses cannot be applied directly.
	By combining the Berry-Esseen approximation for the transmitted-codeword decoding metric in Theorem~\ref{thm:gauss of full D} with the uniform strong large-deviation expansion for the competing-codeword term in Theorem~\ref{thm:uniform-strong-large-deviation},  we first establish a Gaussian approximation directly for the ORB-RCU bound and then use it to derive the third-order achievable rate expansion presented below.
	
	\subsection{Gaussian Approximation of the ORB-RCU Bound}
	\label{subsec:orb-rcu-gaussian-approximation}
	
	Recall from \eqref{rcu n M} that
	\begin{align}
		\operatorname{RCU}_{\mathrm{ORB}}(n,M)
		=
		\Pr\Bigg[
		\ln(M-1)
		+
		\ln F_{\zeta_n}\!\left(
		n^2
		\mathsf D(
		\underline{\mathsf X},
		\underline{\mathsf Y}
		)
		\right)
		-
		\ln\mathsf U
		\ge 0
		\Bigg],
		\label{eq:orb-rcu-log-representation}
	\end{align}
	where $\mathsf U\sim\operatorname{Unif}(0,1)$ is independent of
	all other random variables.
	
	The dependence on the codebook size can be isolated through the
	following normalization. For each integer $M\ge2$, define the
	normalized codebook-size offset
	\begin{align}
		z_{n,M}
		\triangleq
		\frac{
			nI(\mu)
			-
			\ln(M-1)
			+
			\frac{1}{2}\ln n
		}{
			\sqrt n\,|\theta_\mu|\sigma
		},
		\label{eq:normalized-codebook-size-offset}
	\end{align}
	and the normalized random variable
	\begin{align}
		\mathsf S_n
		\triangleq
		\frac{
			nI(\mu)
			+
			\frac{1}{2}\ln n
			+
			\ln F_{\zeta_n}\!\left(
			n^2
			\mathsf D(
			\underline{\mathsf X},
			\underline{\mathsf Y}
			)
			\right)
			-
			\ln\mathsf U
		}{
			\sqrt n\,|\theta_\mu|\sigma
		}.
		\label{eq:normalized-orb-rcu-statistic}
	\end{align}
	
	Combining
	\eqref{eq:orb-rcu-log-representation},
	\eqref{eq:normalized-codebook-size-offset}, and
	\eqref{eq:normalized-orb-rcu-statistic} gives 
	\begin{align}
		\operatorname{RCU}_{\mathrm{ORB}}(n,M)
		=
		\Pr\left[
		\mathsf S_n\ge z_{n,M}
		\right].
		\label{eq:orb-rcu-tail-representation}
	\end{align}
	Note that $\mathsf S_n$ does not depend on the codebook size $M$, which enters only through the deterministic threshold $z_{n,M}$.
	
	It remains to characterize the distribution of $\mathsf S_n$.
	To relate its main random component to the Berry-Esseen
	approximation in Theorem~\ref{thm:gauss of full D}, define the
	deterministic transformation 
	\begin{align}
		\mathcal G_n(t)
		\triangleq
		\frac{
			nI(\mu)
			+
			\frac{1}{2}\ln n
			+
			\ln F_{\zeta_n}\!\left(
			n^2
			\left(
			\mu+\frac{\sigma t}{\sqrt n}
			\right)
			\right)
		}{
			\sqrt n\,|\theta_\mu|\sigma
		},
		\qquad t\in\mathbb R,
		\label{eq:orb-rcu-deterministic-transform}
	\end{align}
	where $\ln 0\triangleq-\infty$. Since
	$F_{\zeta_n}(\cdot)$ is non-decreasing, $\mathcal G_n(\cdot)$ is also non-decreasing. Substituting the normalized
	transmitted-codeword decoding metric into
	\eqref{eq:orb-rcu-deterministic-transform} yields the
	decomposition
	\begin{align}
		\mathsf S_n
		=
		\mathcal G_n\!\left(
		\frac{
			\sqrt n
			\left(
			\mathsf D(
			\underline{\mathsf X},
			\underline{\mathsf Y}
			)
			-
			\mu
			\right)
		}{
			\sigma
		}
		\right)
		-
		\frac{
			\ln\mathsf U
		}{
			\sqrt n\,|\theta_\mu|\sigma
		}.
		\label{eq:orb-rcu-statistic-decomposition}
	\end{align}
	Hence, apart from the independent randomization term involving
	$\mathsf U$, the statistic $\mathsf S_n$ is obtained by applying
	the monotone transformation $\mathcal G_n(\cdot)$ to the normalized transmitted-codeword decoding metric.
	
	The uniform strong large-deviation expansion in
	Theorem~\ref{thm:uniform-strong-large-deviation}, together with
	a second-order expansion of $I(\cdot)$ around $\mu$, shows that
	$\mathcal G_n(\cdot)$ is asymptotically close to the identity
	transformation on the relevant range.
	
	\begin{lemma}\label{lem:near-identity-transform}
		There exist constants $0<C<\infty$ and $n_0\in\mathbb N$ such
		that, for every $n\ge n_0$ and every $t\in\mathbb R$ satisfying
		$\mu+\sigma t/\sqrt n\in\mathcal D$,
		\begin{align}
			\left|
			\mathcal G_n(t)-t
			\right|
			\le
			\frac{C(1+t^2)}{\sqrt n}.
			\label{eq:near-identity-transform}
		\end{align}
	\end{lemma}
	\begin{IEEEproof}
		See Appendix~\ref{app:proof-near-identity-transform}.
	\end{IEEEproof}

	Combining the near-identity property in
	Lemma~\ref{lem:near-identity-transform} with the Berry-Esseen
	approximation in Theorem~\ref{thm:gauss of full D} yields the
	following uniform Gaussian approximation.
	
	\begin{lemma}\label{lem:gaussian-after-transform}
		There exist constants $0<C<\infty$ and $n_0\in\mathbb N$ such
		that, for every $n\ge n_0$,
		\begin{align}
			\sup_{t\in\mathbb R}
			\left|
			\Pr\left[
			\mathcal G_n\!\left(
			\frac{
				\sqrt n
				\left(
				\mathsf D(
				\underline{\mathsf X},
				\underline{\mathsf Y}
				)
				-\mu
				\right)
			}{
				\sigma
			}
			\right)
			\ge t
			\right]
			-
			Q(t)
			\right|
			\le
			\frac{C}{\sqrt n}.
			\label{eq:gaussian-after-transform}
		\end{align}
	\end{lemma}
	\begin{IEEEproof}
		See Appendix~\ref{app:proof-gaussian-after-transform}.
	\end{IEEEproof}
	
	We finally incorporate the independent randomization term involving
	$\mathsf U$. Since this term is of order $n^{-1/2}$ in mean, it
	does not change the order of the Gaussian-approximation error.
	
	\begin{proposition}[Gaussian approximation of the ORB-RCU bound]
		\label{prop:gaussian-approximation-orb-rcu}
		There exist constants $0<C<\infty$ and $n_0\in\mathbb N$ such
		that, for every $n\ge n_0$,
		\begin{align}
			\sup_{t\in\mathbb R}
			\left|
			\Pr[\mathsf S_n\ge t]-Q(t)
			\right|
			\le
			\frac{C}{\sqrt n}.
			\label{eq:gaussian-approximation-Sn}
		\end{align}
		Consequently, uniformly over all $M\ge 2$,
		\begin{align}
			\left|
			\operatorname{RCU}_{\mathrm{ORB}}(n,M)
			-
			Q(z_{n,M})
			\right|
			\le
			\frac{C}{\sqrt n}.
			\label{eq:orb-rcu-global-Gaussian}
		\end{align}
	\end{proposition}
	\begin{IEEEproof}
		Conditioning on $\mathsf U$ in
		\eqref{eq:orb-rcu-statistic-decomposition} gives
		\begin{align}
			\Pr[\mathsf S_n\ge t]
			=
			\int_0^1
			\Pr\Bigg[
			\mathcal G_n\!\left(
			\frac{
				\sqrt n
				\left(
				\mathsf D(
				\underline{\mathsf X},
				\underline{\mathsf Y}
				)
				-\mu
				\right)
			}{
				\sigma
			}
			\right)
			\ge
			t+
			\frac{
				\ln u
			}{
				\sqrt n\,|\theta_\mu|\sigma
			}
			\Bigg]
			\,\mathrm du.
			\label{eq:Sn-conditioning-U}
		\end{align}
		Lemma~\ref{lem:gaussian-after-transform} therefore implies
		that, for some constant $0<C_0<\infty$,
		\begin{align}
			\left|
			\Pr[\mathsf S_n\ge t]
			-
			\int_0^1
			Q\left(
			t+
			\frac{\ln u}{
				\sqrt n\,|\theta_\mu|\sigma
			}
			\right)
			\,\mathrm du
			\right|
			\le
			\frac{C_0}{\sqrt n}
			\label{eq:Sn-after-transform-approximation}
		\end{align}
		uniformly over $t\in\mathbb R$.
		
		The global Lipschitz continuity of $Q(\cdot)$ yields
		\begin{align}
			\left|
			Q\left(
			t+
			\frac{\ln u}{
				\sqrt n\,|\theta_\mu|\sigma
			}
			\right)
			-
			Q(t)
			\right|
			\le
			\frac{-\ln u}{
				\sqrt{2\pi n}\,|\theta_\mu|\sigma
			},
			\qquad 0<u<1.
			\label{eq:Q-randomization-shift}
		\end{align}
		Using $\int_0^1(-\ln u)\,\mathrm du=1$, we obtain
		\begin{align}
			\left|
			\int_0^1
			Q\left(
			t+
			\frac{\ln u}{
				\sqrt n\,|\theta_\mu|\sigma
			}
			\right)
			\,\mathrm du
			-
			Q(t)
			\right|
			\le
			\frac{1}{
				\sqrt{2\pi n}\,|\theta_\mu|\sigma
			}.
			\label{eq:integrated-Q-shift}
		\end{align}
		
		Combining
		\eqref{eq:Sn-after-transform-approximation} and
		\eqref{eq:integrated-Q-shift} proves
		\eqref{eq:gaussian-approximation-Sn}. Finally, evaluating
		\eqref{eq:gaussian-approximation-Sn} at $t=z_{n,M}$ and using
		the exact representation \eqref{eq:orb-rcu-tail-representation} yields \eqref{eq:orb-rcu-global-Gaussian}, uniformly over all $M\ge2$.
	\end{IEEEproof}
	
	Proposition~\ref{prop:gaussian-approximation-orb-rcu} reduces $\mathrm{RCU}_{\mathrm{ORB}}(n,M)$ to a Gaussian tail function evaluated at the deterministic threshold $z_{n,M}$, up to a uniform error of
	order $n^{-1/2}$. This approximation combines the Berry-Esseen approximation for the transmitted-codeword decoding metric with the uniform strong large-deviation expansion for the competing-codeword term. In particular, the former characterizes the Gaussian fluctuations of the transmitted-codeword decoding metric through its nondegenerate second-order $U$-statistic representation, whereas the latter shows that the deterministic transformation $\mathcal G_n(\cdot)$ is asymptotically close to the identity. Crucially, the normalization in $z_{n,M}$ retains the explicit $n^{-1/2}$ subexponential prefactor in the competing-codeword term. After taking logarithms and solving for the codebook size, this prefactor produces the additional $\frac{1}{2}\ln n$ term. We next use the resulting Gaussian approximation to derive the third-order achievable rate expansion for a prescribed average error probability.

	\subsection{Third-Order Achievable Rate Expansion}
	\label{subsec:refined-rate-expansion}
	
	The preceding analysis leads to the following refinement of the threshold-based second-order result in Theorem~\ref{thm:threshold-second-order-rate}.
	
	\begin{theorem}[Third-order achievable rate expansion]
		\label{thm:refined-achievable-rate}
		For every fixed $\epsilon\in(0,1)$, the ORBGRAND
		achievable rate satisfies
		\begin{align}
			R_{\mathrm{ORB}}^{\star}(n,\epsilon)
			\ge
			I(\mu)
			-
			\sqrt{\frac{V_{\mathrm{ORB}}}{n}}\,
			Q^{-1}(\epsilon)
			+
			\frac{\ln n}{2n}
			+
			O\!\left(\frac{1}{n}\right).
			\label{eq:refined-rate-expansion}
		\end{align}
	\end{theorem}
	\begin{IEEEproof}
		Let $0<C_0<\infty$ be the constant in
		\eqref{eq:orb-rcu-global-Gaussian}. For all sufficiently
		large $n$, we have $\epsilon-C_0/\sqrt n\in(0,1)$. Choose the codebook size as
		\begin{align}
			M_n
			\triangleq
			1+
			\left\lfloor
			\exp\left\{
			nI(\mu)
			-
			\sqrt{nV_{\mathrm{ORB}}}\,
			Q^{-1}\!\left(
			\epsilon-\frac{C_0}{\sqrt n}
			\right)
			+
			\frac{1}{2}\ln n
			\right\}
			\right\rfloor .
			\label{eq:achievable-codebook-size}
		\end{align}
		The definition of $M_n$ gives
		\begin{align*}
			\ln(M_n-1)
			\le
			nI(\mu)
			-
			\sqrt{nV_{\mathrm{ORB}}}\,
			Q^{-1}\!\left(
			\epsilon-\frac{C_0}{\sqrt n}
			\right)
			+
			\frac{1}{2}\ln n.
		\end{align*}
		Recalling that
		$\sqrt{V_{\mathrm{ORB}}}=|\theta_\mu|\sigma$, the definition
		of $z_{n,M}$ therefore yields
		\begin{align}
			z_{n,M_n}
			\ge
			Q^{-1}\!\left(
			\epsilon-\frac{C_0}{\sqrt n}
			\right).
			\label{eq:achievable-threshold-bound}
		\end{align}
		The monotonicity of $Q(\cdot)$ and
		\eqref{eq:orb-rcu-global-Gaussian} now give
		\begin{align}
			\operatorname{RCU}_{\mathrm{ORB}}(n,M_n)
			&\le
			Q(z_{n,M_n})
			+
			\frac{C_0}{\sqrt n}
			\le
			Q\!\left(
			Q^{-1}\!\left(
			\epsilon-\frac{C_0}{\sqrt n}
			\right)
			\right)
			+
			\frac{C_0}{\sqrt n}
			=
			\epsilon.
			\label{eq:achievable-codebook-error}
		\end{align}
		The random-coding argument in
		Theorem~\ref{ORB RCU} consequently ensures the existence of
		an $(n,M_n)$ code whose average error probability under
		ORBGRAND does not exceed $\epsilon$. Hence,
		\begin{align}
			M_{\mathrm{ORB}}^{\star}(n,\epsilon)
			\ge M_n.
			\label{eq:optimal-codebook-lower-bound}
		\end{align}
		
		Because $\epsilon\in(0,1)$ is fixed, $Q^{-1}(\cdot)$ is
		continuously differentiable on a neighborhood of $\epsilon$.
		The mean-value theorem thus gives
		\begin{align}
			Q^{-1}\!\left(
			\epsilon-\frac{C_0}{\sqrt n}
			\right)
			=
			Q^{-1}(\epsilon)
			+
			O\!\left(\frac{1}{\sqrt n}\right).
			\label{eq:inverse-Q-local-expansion}
		\end{align}
		Moreover, the elementary inequality
		$1+\lfloor x\rfloor\ge x$ for $x>0$ and
		\eqref{eq:achievable-codebook-size} imply that
		\begin{align}
			\ln M_n
			\ge
			nI(\mu)
			-
			\sqrt{nV_{\mathrm{ORB}}}\,
			Q^{-1}\!\left(
			\epsilon-\frac{C_0}{\sqrt n}
			\right)
			+
			\frac{1}{2}\ln n.
			\label{eq:achievable-codebook-log-bound}
		\end{align}
		Combining
		\eqref{eq:optimal-codebook-lower-bound},
		\eqref{eq:inverse-Q-local-expansion}, and
		\eqref{eq:achievable-codebook-log-bound}, we obtain
		\begin{align}
			\ln M_{\mathrm{ORB}}^{\star}(n,\epsilon)
			\ge
			nI(\mu)
			-
			\sqrt{nV_{\mathrm{ORB}}}\,
			Q^{-1}(\epsilon)
			+
			\frac{1}{2}\ln n
			+
			O(1).
			\label{eq:proved-codebook-size-expansion}
		\end{align}
		Dividing both sides by $n$ and using the definition of
		$R_{\mathrm{ORB}}^{\star}(n,\epsilon)$ proves
		\eqref{eq:refined-rate-expansion}.
	\end{IEEEproof}
	
	Theorem~\ref{thm:threshold-second-order-rate} establishes
	the first- and second-order terms of the achievable rate,
	whereas Theorem~\ref{thm:refined-achievable-rate} further
	retains the positive correction $\frac{\ln n}{2n}$. In both results, $I(\mu)$ appears as the first-order term, while
	$V_{\mathrm{ORB}}$ determines the second-order backoff. We refer
	to $V_{\mathrm{ORB}}$ as the ORBGRAND dispersion. The
	additional logarithmic correction in
	Theorem~\ref{thm:refined-achievable-rate} originates from the
	$n^{-1/2}$ prefactor in the strong large-deviation expansion
	of $F_{\zeta_n}(\cdot)$.
	
	The second- and third-order expansions in
	\eqref{eq:threshold-second-order-achievable-rate} and
	\eqref{eq:refined-rate-expansion} both contain the leading term
	$I(\mu)$, obtained by evaluating the large-deviation rate
	function at the typical value $\mu$ of the transmitted-codeword decoding metric. On the other hand, previous asymptotic analyses
	characterized the achievable rate of ORBGRAND from the
	perspective of generalized mutual information (GMI)
	\cite{liu2022orbgrand,li2026orbgrand}. The following result shows
	that $I(\mu)$ coincides with the achievable rate expression
	established therein, thereby connecting the present
	finite-blocklength analysis with the existing asymptotic characterization.

	\begin{proposition}[Identification of the first-order rate]
		\label{prop:I_mu_Iorb_new}
		The rate-function value $I(\mu)$ coincides with the
		ORBGRAND achievable rate $I_{\mathrm{ORB}}$ in
		\cite[(6)]{li2026orbgrand}, i.e.,
		\begin{align}
			I(\mu)=I_{\mathrm{ORB}}.
			\label{eq:I-mu-equals-I-orb}
		\end{align}
	\end{proposition}
	
	\begin{IEEEproof}
		Recall that $K'(0)=1/4$. Since $K'(\cdot)$ is strictly
		increasing and $\mu\in(0,1/4)$, the unique solution to
		$K'(\theta_\mu)=\mu$ satisfies $\theta_\mu<0$. Hence, the
		unique maximizer in the definition of $I(\mu)$ lies in
		$(-\infty,0)$, and
		\begin{align}
			I(\mu)
			=
			\sup_{\theta<0}
			\left\{
			\theta\mu-K(\theta)
			\right\}.
			\label{eq:I-mu-negative-theta}
		\end{align}
		Moreover, \eqref{eq:K} gives
		\begin{align}
			K(\theta)
			=
			\int_0^1
			\ln\left(1+e^{\theta t}\right)
			\,\mathrm{d}t
			-
			\ln 2.
			\label{eq:K-alternative-form}
		\end{align}
		
		Recalling \eqref{eq:def_E_mu}, the quantity $\mu$ can be
		written as
		\begin{align}
			\mu
			&=
			\frac{1}{2}
			\int_{q^+(y)<q^-(y)}
			\Psi\left(
			\left|
			\ln\frac{q^+(y)}{q^-(y)}
			\right|
			\right)
			q^+(y)
			\,\mathrm{d}y
			+
			\frac{1}{2}
			\int_{q^+(y)>q^-(y)}
			\Psi\left(
			\left|
			\ln\frac{q^+(y)}{q^-(y)}
			\right|
			\right)
			q^-(y)
			\,\mathrm{d}y.
			\label{eq:mu-two-decision-regions}
		\end{align}
		Substituting
		\eqref{eq:K-alternative-form} and
		\eqref{eq:mu-two-decision-regions} into
		\eqref{eq:I-mu-negative-theta} yields
		\begin{align}
			I(\mu)
			&=
			\ln 2
			-
			\inf_{\theta<0}
			\Bigg\{
			\int_0^1
			\ln\left(1+e^{\theta t}\right)
			\,\mathrm{d}t
			\nonumber\\
			&\qquad
			-
			\frac{\theta}{2}
			\int_{q^+(y)<q^-(y)}
			\Psi\left(
			\left|
			\ln\frac{q^+(y)}{q^-(y)}
			\right|
			\right)
			q^+(y)
			\,\mathrm{d}y
			\nonumber\\
			&\qquad
			-
			\frac{\theta}{2}
			\int_{q^+(y)>q^-(y)}
			\Psi\left(
			\left|
			\ln\frac{q^+(y)}{q^-(y)}
			\right|
			\right)
			q^-(y)
			\,\mathrm{d}y
			\Bigg\}.
			\label{eq:I-mu-I-orb-expression}
		\end{align}
		The right-hand side of
		\eqref{eq:I-mu-I-orb-expression} is precisely the expression
		for $I_{\mathrm{ORB}}$ in
		\cite[(6)]{li2026orbgrand}, thereby proving
		\eqref{eq:I-mu-equals-I-orb}.
	\end{IEEEproof}
	
	Combining Proposition~\ref{prop:I_mu_Iorb_new} with Theorems~\ref{thm:threshold-second-order-rate} and~\ref{thm:refined-achievable-rate}, respectively, yields
	\begin{align}
		R_{\mathrm{ORB}}^{\star}(n,\epsilon)
		&\ge
		I_{\mathrm{ORB}}
		-
		\sqrt{\frac{V_{\mathrm{ORB}}}{n}}\,
		Q^{-1}(\epsilon)
		+
		O\!\left(\frac{1}{n}\right),
		\label{eq:orb-second-order-achievable-rate}\\
		R_{\mathrm{ORB}}^{\star}(n,\epsilon)
		&\ge
		I_{\mathrm{ORB}}
		-
		\sqrt{\frac{V_{\mathrm{ORB}}}{n}}\,
		Q^{-1}(\epsilon)
		+
		\frac{\ln n}{2n}
		+
		O\!\left(\frac{1}{n}\right).
		\label{eq:orb-refined-achievable-rate}
	\end{align}
	The leading term $I_{\mathrm{ORB}}$ recovers the asymptotic
	achievable rate of ORBGRAND established in
	\cite{liu2022orbgrand,li2026orbgrand}. The ORBGRAND dispersion
	$V_{\mathrm{ORB}}$ governs the second-order rate backoff. The
	first bound retains the asymptotic and second-order terms,
	whereas the second additionally includes the logarithmic
	third-order correction $\frac{\ln n}{2n}$. Accordingly, neglecting
	the remainder terms in
	\eqref{eq:orb-second-order-achievable-rate} and
	\eqref{eq:orb-refined-achievable-rate}, respectively, motivates
	the second- and third-order approximations:
	\begin{align}
		\widehat{R}_{\mathrm{ORB}}^{(2)}(n,\epsilon)
		&\triangleq
		I_{\mathrm{ORB}}
		-
		\sqrt{\frac{V_{\mathrm{ORB}}}{n}}\,
		Q^{-1}(\epsilon),
		\label{eq:orb-second-order-approximation}\\
		\widehat{R}_{\mathrm{ORB}}^{(3)}(n,\epsilon)
		&\triangleq
		I_{\mathrm{ORB}}
		-
		\sqrt{\frac{V_{\mathrm{ORB}}}{n}}\,
		Q^{-1}(\epsilon)
		+
		\frac{\ln n}{2n}.
		\label{eq:orb-third-order-approximation}
	\end{align}
	
	\begin{remark}[Relation to classical finite-blocklength analyses]
		\label{rem:decoder-induced-NA}
		The second-order approximation in
		\eqref{eq:orb-second-order-approximation} has the same formal
		structure as the classical normal approximation for
		finite-blocklength channel coding
		\cite{polyanskiy2010channel} and its extensions to mismatched
		decoding with independent single-letter decoding metrics \cite{scarlett2014mismatched}. The third-order approximation in
		\eqref{eq:orb-third-order-approximation} further retains the
		logarithmic correction $\frac{\ln n}{2n}$. The quantities appearing in both approximations, however, are induced by the ORBGRAND rule. Specifically, $I_{\mathrm{ORB}}$ is the first-order achievable rate associated with
		reliability-ordered guessing, while
		$V_{\mathrm{ORB}}=\theta_\mu^2\sigma^2$ characterizes the
		second-order fluctuation of the transmitted-codeword decoding
		metric.
		
		Unlike the matched information density or independent single-letter decoding metrics used in mismatched decoding, the ORBGRAND decoding metrics are rank-based and coupled across symbols through the global reliability ordering. Consequently, neither the
		second-order nor the third-order result follows directly from
		a normal approximation for a sum of independent
		single-letter random variables. Their derivations instead have relied on two ingredients: the Berry-Esseen approximation in Theorem~\ref{thm:gauss of full D}, obtained through the nondegenerate second-order $U$-statistic representation of the
		transmitted-codeword decoding metric, and the uniform strong
		large-deviation expansion in
		Theorem~\ref{thm:uniform-strong-large-deviation}, which provides a refined characterization of the term $F_{\zeta_n}(\cdot)$ associated with the competing codeword.
		
		The threshold-based analysis separates the transmitted- and
		competing-codeword contributions and yields the second-order
		approximation in \eqref{eq:orb-second-order-approximation}. A direct analysis of the ORB-RCU bound further retains the conditional pairwise probability within the expression in \eqref{eq:rcu_n_M}, thereby
		preserving the explicit $n^{-1/2}$ subexponential prefactor in
		the uniform strong large-deviation expansion of the term $F_{\zeta_n}(\cdot)$ associated with the competing codeword. Taking logarithms and solving for the codebook size then produces the
		additional $\frac{1}{2}\ln n$ term, or equivalently the
		correction $\frac{\ln n}{2n}$ in
		\eqref{eq:orb-third-order-approximation}.
	\end{remark}

	\section{Numerical Results}\label{sec:numerical}
	
	In this section, we present numerical results for the BPSK-modulated AWGN channel to validate the obtained finite-blocklength characterizations.
	We consider two complementary viewpoints: \emph{(i)} error probability curves at a fixed blocklength, and \emph{(ii)} rate-blocklength curves at a fixed target error probability. To provide benchmarks, we compute the meta-converse bound using the saddlepoint approximation for binary hypothesis testing in \cite{vazquez2018saddlepoint}, and evaluate the RCU bound under ML decoding using the saddlepoint approximations in \cite{font2018saddlepoint}.
	
	\subsection{Error Probability versus Rate at Fixed Blocklength}
	
	For a given signal-to-noise ratio (SNR) and blocklength $n$,
	we plot, as functions of the rate $R$, the meta-converse bound,
	the ML-RCU bound, the ORB-RCU bound in
	Theorem~\ref{ORB RCU}, and the second- and third-order ORB
	approximations obtained by inverting
	\eqref{eq:orb-second-order-approximation} and
	\eqref{eq:orb-third-order-approximation}, respectively.
	The meta-converse provides a lower bound on the minimum average
	error probability, whereas the ML-RCU and ORB-RCU bounds provide
	upper bounds on the average error probability under ML decoding
	and ORBGRAND, respectively. The second- and third-order curves, on the other hand, are obtained by inverting the corresponding rate expansions and provide numerical approximations to the average error probability under ORBGRAND rather than error probability bounds.

	\begin{figure}[htbp]	\centerline{\includegraphics[width=0.75\textwidth]{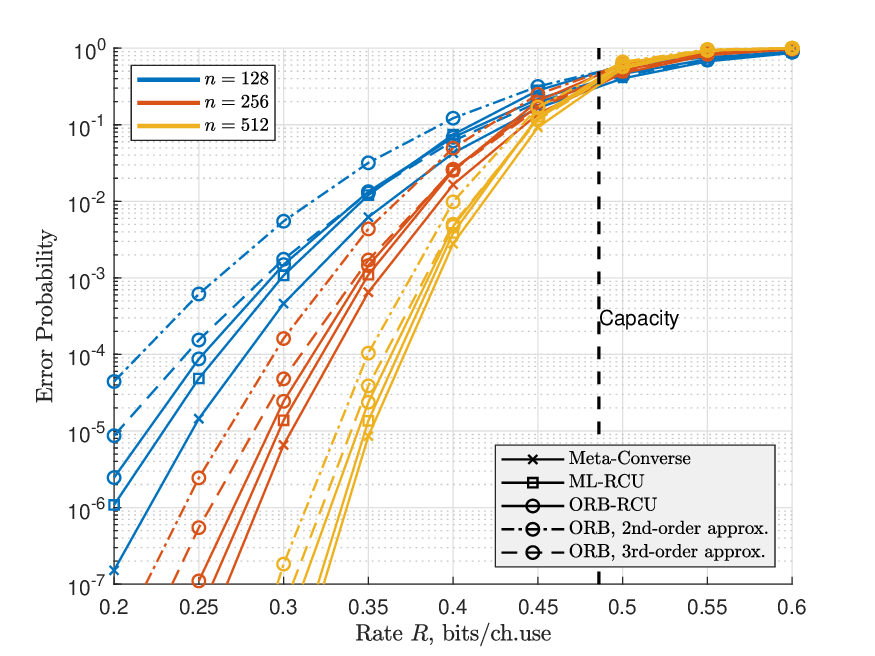}}
		\caption{Average error probability bounds versus rate $R$ for the BPSK-modulated AWGN channel at $\mathrm{SNR}=0$~dB.}
		\label{fig:eps_vs_R0dB}
	\end{figure}
	
	\begin{figure}[htbp]	\centerline{\includegraphics[width=0.75\textwidth]{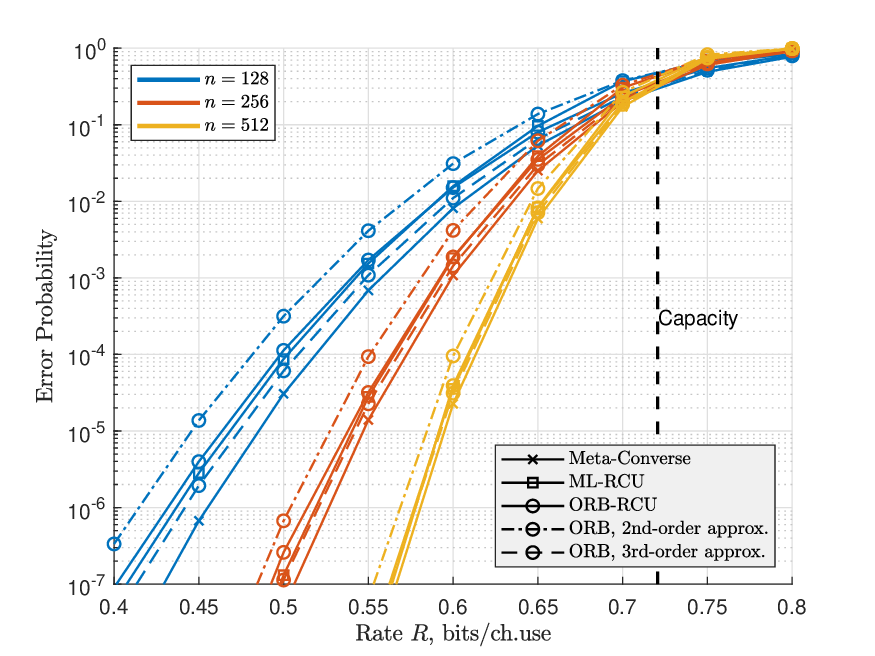}}
		\caption{Average error probability bounds versus rate $R$ for the BPSK-modulated AWGN channel at $\mathrm{SNR}=3$~dB.}
		\label{fig:eps_vs_R3dB}
	\end{figure}

	Figs.~\ref{fig:eps_vs_R0dB} and
	\ref{fig:eps_vs_R3dB} show that, for both SNR values therein, the
	ORB-RCU curve closely tracks the ML-RCU benchmark over the
	considered rate range. This indicates that the finite-blocklength error probability penalty of ORBGRAND relative to
	ML decoding is modest in the regimes considered. Moreover, at a fixed coding rate, the error probability gap
	between the ORB-RCU and ML-RCU curves becomes smaller as the
	blocklength increases, suggesting that the error probability
	penalty of ORBGRAND relative to ML decoding becomes less
	pronounced at larger blocklengths. The second- and third-order ORB error probability approximations both capture the overall dependence of the ORBGRAND error probability on the coding rate. In particular, the third-order approximation generally follows the ORB-RCU curve more closely than the second-order approximation in the operating region of interest, illustrating the improvement provided by the logarithmic correction $\frac{\ln n}{2n}$. As expected, the meta-converse remains below the achievability bounds.\footnote{We observed occasional crossings where the ORB-RCU curve slightly falls below the ML-RCU curve at a few rate points. This does not contradict the optimality of ML decoding, since both curves are achievability upper bounds obtained via different relaxations and truncations (e.g., the $\min\{1,\cdot\}$ clipping in RCU-type bounds); hence they are not guaranteed to be pointwise ordered.}

	\subsection{Rate versus Blocklength at Fixed Target Error Probability}
	
	We now fix the SNR and target average error probability $\epsilon$. For each blocklength $n$, the rate curves corresponding to the meta-converse, ML-RCU, and ORB-RCU bounds
	are obtained by numerically inverting the corresponding
	finite-blocklength relations to determine the code size $M$.
	The second- and third-order ORB approximations curves are instead obtained by directly evaluating \eqref{eq:orb-second-order-approximation} and
	\eqref{eq:orb-third-order-approximation}, respectively.

	\begin{figure}[t]	\centerline{\includegraphics[width=0.75\textwidth]{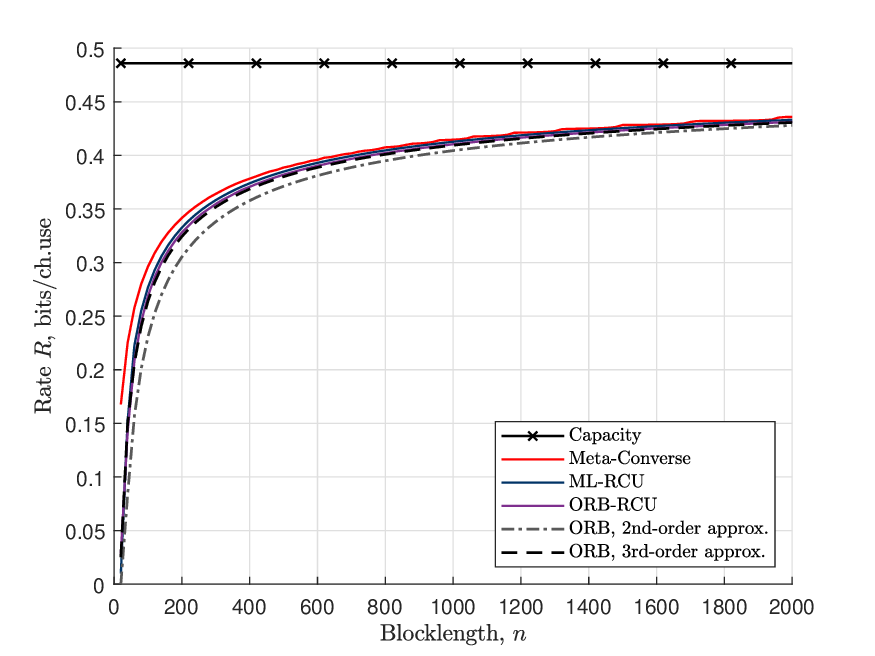}}
		\caption{Rate-blocklength curves for the BPSK-modulated AWGN channel at $\mathrm{SNR}=0$~dB and $\epsilon=10^{-3}$.}
		\label{fig:0dBrate}
	\end{figure}
	\begin{figure}[t]	\centerline{\includegraphics[width=0.75\textwidth]{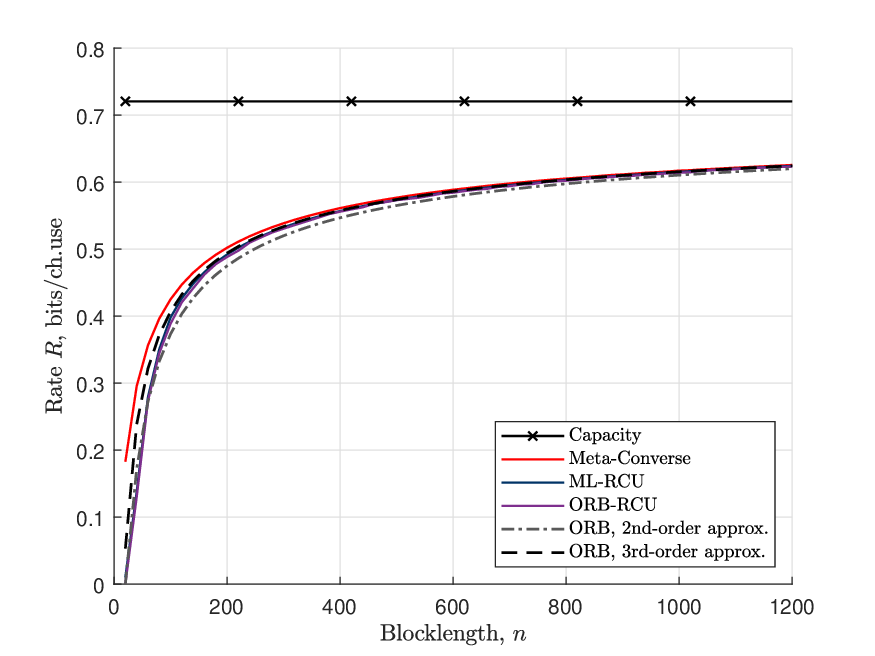}}
		\caption{Rate-blocklength curves for the BPSK-modulated AWGN channel at $\mathrm{SNR}=3$~dB and $\epsilon=10^{-6}$.}
		\label{fig:2dBrate}
	\end{figure}
	
	Figs.~\ref{fig:0dBrate} and~\ref{fig:2dBrate} report the resulting rate-blocklength curves. For both SNR values therein, the ORB-RCU curve stays close to the ML-RCU benchmark over the considered blocklengths, indicating only a small finite-blocklength rate penalty for ORBGRAND relative to ML. The second-order ORB approximation already captures the overall finite-blocklength trend of the ORB-RCU curve. Incorporating the logarithmic correction $\frac{\ln n}{2n}$ likewise improves the accuracy, with the third-order ORB approximation closely tracking the ORB-RCU curve even at moderate blocklengths, such as $n = 100$, and becoming increasingly accurate as $n$ grows. These finite-blocklength observations are qualitatively
	consistent with the asymptotic trend reported
	in~\cite{liu2022orbgrand}. In particular,
	while~\cite{liu2022orbgrand} establishes that ORBGRAND
	approaches the channel capacity as $n\to\infty$,
	our results indicate that similar near-optimal behavior is
	already visible at practical blocklengths.
	
	\begin{table}[H]
		\caption{Bounds on the minimal blocklength $n$ needed to achieve $R$ over the BPSK-modulated AWGN channel.}
		\label{tab:min_n_bounds}
		\centering
		\renewcommand{\arraystretch}{1.0}
		\resizebox{\columnwidth}{!}{
			\begin{tabular}{c c c c c c c c}
				\hline
				\textbf{Rate}
				& \textbf{SNR (dB)}
				& $\boldsymbol{\epsilon}$
				& \textbf{Meta-Converse}
				& \textbf{ML-RCU}
				& \textbf{ORB-RCU}
				& \textbf{ORB 2nd-Order}
				& \textbf{ORB 3rd-Order}
				\\
				\hline
				\multirow{4}{*}{$0.8C$}
				& $0$ & $10^{-3}$
				& $n\ge 503$
				& $n\le 545$
				& $n\le 579$
				& $n\approx 700$
				& $n\approx 600$
				\\
				& $1$ & $10^{-4}$
				& $n\ge 565$
				& $n\le 599$
				& $n\le 620$
				& $n\approx 723$
				& $n\approx 638$
				\\
				& $2$ & $10^{-5}$
				& $n\ge 556$
				& $n\le 582$
				& $n\le 592$
				& $n\approx 671$
				& $n\approx 598$
				\\
				& $3$ & $10^{-6}$
				& $n\ge 496$
				& $n\le 522$
				& $n\le 532$
				& $n\approx 582$
				& $n\approx 518$
				\\
				\hline
				\multirow{4}{*}{$0.9C$}
				& $0$ & $10^{-3}$
				& $n\ge 2272$
				& $n\le 2382$
				& $n\le 2585$
				& $n\approx 2874$
				& $n\approx 2629$
				\\
				& $1$ & $10^{-4}$
				& $n\ge 2427$
				& $n\le 2578$
				& $n\le 2656$
				& $n\approx 2919$
				& $n\approx 2708$
				\\
				& $2$ & $10^{-5}$
				& $n\ge 2268$
				& $n\le 2472$
				& $n\le 2512$
				& $n\approx 2693$
				& $n\approx 2517$
				\\
				& $3$ & $10^{-6}$
				& $n\ge 2031$
				& $n\le 2181$
				& $n\le 2220$
				& $n\approx 2343$
				& $n\approx 2187$
				\\
				\hline
			\end{tabular}
		}
	\end{table}
	
	In addition, we fix the operating rate to $R=0.8C$ and
	$R=0.9C$, where $C$ denotes the channel capacity under the
	BPSK input, and for each target error probability $\epsilon$
	numerically invert the meta-converse, ML-RCU,
	and ORB-RCU bounds to characterize the blocklength required
	to achieve the pair $(R,\epsilon)$. The meta-converse yields a
	necessary lower bound on the minimal blocklength
	$n^{\star}(R,\epsilon)$, whereas ML-RCU and ORB-RCU provide
	sufficient upper bounds.
	For comparison, we also report the
	blocklength estimates obtained by numerically inverting the
	second- and third-order ORB approximations in
	\eqref{eq:orb-second-order-approximation} and
	\eqref{eq:orb-third-order-approximation}, respectively. These
	approximations provide numerical estimates rather than
	rigorous bounds.
	The results are summarized in Table~\ref{tab:min_n_bounds}.
	The blocklength upper bound obtained from
	ORB-RCU is generally slightly larger than that obtained from
	ML-RCU, reflecting the performance loss of ORBGRAND relative
	to ML decoding. The second-order approximation captures the overall blocklength trend, while the logarithmic correction $\frac{\ln n}{2n}$ yields a more accurate prediction. In particular, the third-order approximation shows good agreement with the ORB-RCU results across the considered operating points, once again demonstrating the benefit of the
	logarithmic correction for finite-blocklength prediction.
	
	\begin{figure}[H]	\centerline{\includegraphics[width=0.75\textwidth]{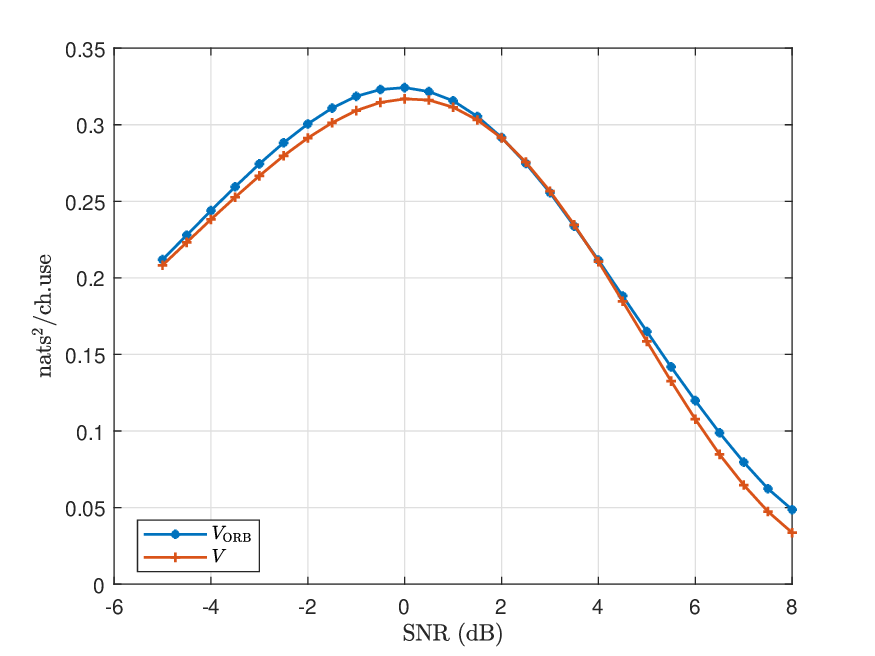}}
		\caption{$V_{\mathrm{ORB}}$ and $V$ versus SNR for the BPSK-modulated AWGN channel.}
		\label{fig:V_ORB}
	\end{figure} 
	
	Fig.~\ref{fig:V_ORB} plots the ORBGRAND dispersion $V_{\mathrm{ORB}}$ together with the ML dispersion $V$ as functions of SNR for the BPSK-modulated AWGN channel. Both curves exhibit a clear unimodal behavior: the dispersion increases from the low-SNR regime, peaks around $\text{SNR}\approx 0$~dB, and then decreases steadily as SNR grows. Moreover, $V_{\mathrm{ORB}}$ closely tracks $V$ over the entire SNR range, indicating that ORBGRAND induces a second-order behavior that is comparable to that of optimal ML decoding on the AWGN channel.
	
	This SNR dependence is directly reflected in the finite-blocklength achievable rate expansions in  \eqref{eq:orb-second-order-achievable-rate} and \eqref{eq:orb-refined-achievable-rate}. Since $\sqrt{V_{\mathrm{ORB}}}$ is the dispersion coefficient multiplying the $1/\sqrt{n}$ term in the finite-blocklength achievable rate expansions, it governs the finite-blocklength backoff $\sqrt{V_{\mathrm{ORB}}/n}\,Q^{-1}(\epsilon)$. Consequently, the peak of $V_{\mathrm{ORB}}$ around moderate SNR implies that the second-order backoff from the first-order term $I_{\mathrm{ORB}}$ is most pronounced in this regime, whereas at higher SNR the dispersion decreases and the second-order penalty becomes smaller, so the achievable rate approaches $I_{\mathrm{ORB}}$ more rapidly as $n$ increases.

	\section{Conclusion}\label{sec:con}
	
	This paper established a decoder-dependent finite-blocklength
	analysis for ORBGRAND over general binary-input memoryless channels. Since ORBGRAND is driven by reliability ordering, the induced decoding metrics are rank-based and coupled across symbols, and standard finite-blocklength results are
	not directly applicable. Our analysis begins by specializing the general mismatched RCU bound to the rank-based ORBGRAND rule, resulting in an ensemble-average error probability bound that we refer to as the ORB-RCU bound.
	
	We further characterized the asymptotic behavior of
	ORBGRAND through a rank-statistic analysis:
	the transmitted-codeword decoding metric
	admits a nondegenerate second-order $U$-statistic
	representation, which leads to a Berry-Esseen approximation,
	whereas the competing-codeword term is characterized
	through a uniform strong large-deviation expansion. Combining
	these two ingredients first yields a threshold-based
	second-order achievable rate expansion. A direct analysis of
	the ORB-RCU bound further preserves the
	large-deviation prefactor and yields a refined third-order
	expansion containing the logarithmic correction
	$\frac{\ln n}{2n}$. We also showed that the first-order term
	$I(\mu)$ coincides with the previously established ORBGRAND
	achievable rate $I_{\mathrm{ORB}}$ in
	\cite[(6)]{li2026orbgrand}, thereby connecting the present finite-blocklength analysis with the existing asymptotic characterization and yielding decoder-induced second- and third-order approximations.
	
	Numerical results on the BPSK-modulated AWGN channel corroborate the theoretical analysis:
	ORB-RCU closely tracks the ML-RCU benchmark over the
	considered operating regime, while the third-order
	approximation closely follows ORB-RCU and provides a more
	accurate finite-blocklength prediction than the second-order
	approximation, particularly at short and moderate
	blocklengths. These results quantify the finite-blocklength gap between ORBGRAND and ML decoding and provide practical design
	guidelines without relying solely on exhaustive simulations.

	Finally, two directions appear particularly
	worthy of further investigation. Motivated by recent work on
	guessing-based decoders with abandonment \cite{tan2025ensemble}, one direction is to impose truncation on the maximum number of queries for ORBGRAND and characterize the resulting finite blocklength
	rate-reliability-complexity tradeoff. Another direction is to
	develop a converse tailored to ORBGRAND tighter than the meta-converse, which would provide a
	useful complement to the achievability results established in
	this work.

	\section{Appendices}\label{sec:app}
	
	\subsection{Proof of Lemma \ref{conv point}}\label{proof conv point}
	
	Since $\Psi(\Lambda)$ is a deterministic function of
	$\mathsf Y$, the law of total expectation gives
	\begin{align}
		\mu
		&=
		\mathbb E[\Psi(\Lambda)\mathsf E]
		=
		\mathbb E\!\left[
		\Psi(\Lambda)
		\Pr[\mathsf E=1\mid\mathsf Y]
		\right].
	\end{align}
	For the LLR under the equiprobable input distribution, the
	conditional hard-decision error probability satisfies
	\begin{align}
		\Pr[\mathsf E=1\mid\mathsf Y=y]
		=
		\frac{1}{1+e^{|l(y)|}}.
	\end{align}
	Therefore, we have
	\begin{align}
		\mu
		=
		\mathbb E\!\left[
		\frac{\Psi(\Lambda)}{1+e^\Lambda}
		\right].
		\label{eq:mu-representation}
	\end{align}
	
	By Assumption~\ref{assump:regularity}, $\Lambda$ is continuously
	distributed. Hence, the probability integral transform gives
	$\Psi(\Lambda)\sim\operatorname{Unif}[0,1]$, and consequently
	\begin{align}
		\Psi(\Lambda)>0
		\quad\text{almost surely},
		\qquad
		\mathbb E[\Psi(\Lambda)]=\frac12.
	\end{align}
	Moreover, since $\Lambda\ge 0$ and
	$\Pr[\Lambda=0]=0$, we have $\Lambda>0$ almost surely. It follows
	that
	\begin{align}
		0
		<
		\frac{\Psi(\Lambda)}{1+e^\Lambda}
		<
		\frac12\Psi(\Lambda)
		\quad\text{almost surely}.
	\end{align}
	Taking expectations and using \eqref{eq:mu-representation}, we
	obtain
	\begin{align}
		0
		<
		\mu
		<
		\frac12\mathbb E[\Psi(\Lambda)]
		=
		\frac14.
	\end{align}

	\subsection{Proof of Theorem \ref{thm:gauss of full D}}\label{proof gauss full D}

	Let $\mathsf{Z}_i\triangleq(\Lambda_i,\mathsf{E}_i)$, and define the symmetric kernel
	\begin{align}
		h\bigl((\lambda,e),(\lambda',e')\bigr)
		\triangleq
		\frac{1}{2}
		\left[
		e\mathbf{1}(\lambda'\le\lambda)
		+
		e'\mathbf{1}(\lambda\le\lambda')
		\right].
	\end{align}
	The corresponding second-order $U$-statistic \cite{hoeffding1992class} is
	\begin{align}
		\mathsf{V}_n
		\triangleq
		\binom{n}{2}^{-1}
		\sum_{1\le i<j\le n}
		h(\mathsf{Z}_i,\mathsf{Z}_j).
	\end{align}
	Using the definition of the reliability ranks,
	\begin{align}
		\mathsf R_i
		=
		\sum_{j=1}^{n}
		\mathbf 1(\Lambda_j\le\Lambda_i),
	\end{align}
	we have
	\begin{align}
		\mathsf D(\underline{\mathsf X},\underline{\mathsf Y})
		&=
		\frac{1}{n^2}
		\sum_{i=1}^{n}
		\mathsf E_i
		\sum_{j=1}^{n}
		\mathbf 1(\Lambda_j\le\Lambda_i)
		\nonumber\\
		&=
		\frac{1}{n^2}
		\sum_{i=1}^{n}\mathsf E_i
		+
		\frac{1}{n^2}
		\sum_{\substack{i,j=1\\ i\ne j}}^{n}
		\mathsf E_i
		\mathbf 1(\Lambda_j\le\Lambda_i)
		\nonumber\\
		&=
		\frac{1}{n^2}
		\sum_{i=1}^{n}\mathsf E_i
		+
		\frac{1}{n^2}
		\sum_{1\le i<j\le n}
		\Bigl[
		\mathsf E_i\mathbf 1(\Lambda_j\le\Lambda_i)
		+
		\mathsf E_j\mathbf 1(\Lambda_i\le\Lambda_j)
		\Bigr]
		\nonumber\\
		&=
		\frac{1}{n^2}
		\sum_{i=1}^{n}\mathsf E_i
		+
		\frac{2}{n^2}
		\sum_{1\le i<j\le n}
		h(\mathsf Z_i,\mathsf Z_j)
		\nonumber\\
		&=
		\frac{1}{n^2}
		\sum_{i=1}^{n}\mathsf E_i
		+
		\frac{2}{n^2}
		\binom{n}{2}\mathsf V_n
		\nonumber\\
		&=
		\left(1-\frac{1}{n}\right)\mathsf V_n
		+
		\frac{1}{n^2}
		\sum_{i=1}^{n}\mathsf E_i.
		\label{eq:D-U-statistic-decomposition}
	\end{align}
	Since $0\le\mathsf{V}_n\le1$ and
	$0\le\sum_{i=1}^{n}\mathsf{E}_i\le n$, it follows that
	\begin{align}\label{eq:D_Mn_bound}
		\left|
		\mathsf{D}(\underline{\mathsf{X}},\underline{\mathsf{Y}})
		-\mathsf{V}_n
		\right|
		\le
		\frac{1}{n}.
	\end{align}
	
	Let $\mathsf Z'=(\Lambda',\mathsf E')$ be an independent copy
	of $\mathsf Z=(\Lambda,\mathsf E)$. Since $\mathsf V_n$ is a
	second-order $U$-statistic with kernel $h(\cdot,\cdot)$, its expectation
	coincides with the mean of the kernel. Using the definition of
	$h(\cdot,\cdot)$, we obtain
	\begin{align}
		\mathbb E[\mathsf V_n]
		&=
		\mathbb E\bigl[h(\mathsf Z,\mathsf Z')\bigr]
		\nonumber\\
		&=
		\frac{1}{2}
		\mathbb E\left[
		\mathsf E
		\mathbf 1(\Lambda'\le\Lambda)
		\right]
		+
		\frac{1}{2}
		\mathbb E\left[
		\mathsf E'
		\mathbf 1(\Lambda\le\Lambda')
		\right]
		\nonumber\\
		&=
		\frac{1}{2}
		\mathbb E\left[
		\mathsf E\Psi(\Lambda)
		\right]
		+
		\frac{1}{2}
		\mathbb E\left[
		\mathsf E'\Psi(\Lambda')
		\right]
		\nonumber\\
		&=
		\mathbb E\left[
		\mathsf E\Psi(\Lambda)
		\right]
		=
		\mu.
		\label{eq:kernel-expectation}
	\end{align}
	Here, the third equality follows from the independence of
	$\mathsf Z$ and $\mathsf Z'$, while the fourth equality follows
	because $(\Lambda',\mathsf E')$ has the same distribution as
	$(\Lambda,\mathsf E)$.
	
	Conditioning on $\mathsf Z$, we have
	\begin{align}
		\mathbb E\left[
		h(\mathsf Z,\mathsf Z')
		\mid\mathsf Z
		\right]
		&=
		\frac{1}{2}
		\mathsf E\,
		\mathbb E\left[
		\mathbf 1(\Lambda'\le\Lambda)
		\mid\mathsf Z
		\right]
		+
		\frac{1}{2}
		\mathbb E\left[
		\mathsf E'
		\mathbf 1(\Lambda\le\Lambda')
		\mid\mathsf Z
		\right]
		\nonumber\\
		&=
		\frac{1}{2}
		\left[
		\mathsf E\Psi(\Lambda)
		+
		\Pr[
		\mathsf E'=1,\Lambda'\ge\Lambda
		\mid\mathsf Z
		]
		\right]
		\nonumber\\
		&=
		\frac{1}{2}
		\left[
		\mathsf E\Psi(\Lambda)+a(\Lambda)
		\right],
		\label{eq:first-order-conditional-expectation}
	\end{align}
	where the last equality follows from the independence of
	$\mathsf Z'$ and $\mathsf Z$ and the definition of $a(x)$.
	Therefore, the first-order Hoeffding projection \cite{hoeffding1992class} is
	\begin{align}
		\mathbb E\left[
		h(\mathsf Z,\mathsf Z')
		\mid\mathsf Z
		\right]
		-
		\mathbb E\left[
		h(\mathsf Z,\mathsf Z')
		\right]
		=
		\frac{1}{2}
		\left[
		\mathsf E\Psi(\Lambda)+a(\Lambda)
		\right]
		-
		\mu.
	\end{align}
	Subtracting the constant $\mu$ does not affect the variance.
	Consequently,
	\begin{align}
		\operatorname{Var}\left(
		\mathbb E\left[
		h(\mathsf Z,\mathsf Z')
		\mid\mathsf Z
		\right]
		\right)
		&=
		\frac{1}{4}
		\operatorname{Var}\left(
		\mathsf E\Psi(\Lambda)+a(\Lambda)
		\right)
		=
		\frac{\sigma^2}{4}.
		\label{eq:first-order-projection-variance}
	\end{align}
	
	We next verify that the first-order projection is nondegenerate. Since $\mathsf L$ is the LLR under the
	equiprobable input distribution, we have
	\begin{align}
		\Pr[\mathsf E=1\mid\Lambda=\lambda]
		=
		\frac{1}{1+e^\lambda}
	\end{align}
	for almost every $\lambda$ with respect to the distribution of
	$\Lambda$. It follows that the corresponding conditional
	variance is
	\begin{align}
		\operatorname{Var}(\mathsf E\mid\Lambda=\lambda)
		=
		\frac{e^\lambda}{(1+e^\lambda)^2}.
	\end{align}
	Moreover, conditioned on $\Lambda=\lambda$, both
	$\Psi(\Lambda)$ and $a(\Lambda)$ are deterministic. Therefore,
	\begin{align}
		&
		\operatorname{Var}\!\left(
		\mathsf E\Psi(\Lambda)+a(\Lambda)
		\mid\Lambda=\lambda
		\right)
		=
		\Psi(\lambda)^2
		\frac{e^\lambda}{(1+e^\lambda)^2}.
		\label{eq:conditional-projection-variance}
	\end{align}
	Applying the law of total variance to
	$\mathsf E\Psi(\Lambda)+a(\Lambda)$ now gives
	\begin{align}
		\sigma^2
		&\ge
		\mathbb E\!\left[
		\Psi(\Lambda)^2
		\frac{e^\Lambda}{(1+e^\Lambda)^2}
		\right].
		\label{eq:positive-projection-variance}
	\end{align}
	Under Assumption~\ref{assump:regularity},
	$\Psi(\Lambda)\sim\operatorname{Unif}[0,1]$ and is
	therefore strictly positive almost surely. Since
	$e^\Lambda/(1+e^\Lambda)^2$ is also strictly positive almost
	surely, the expectation in
	\eqref{eq:positive-projection-variance} is strictly positive.
	Consequently, $\sigma^2>0$.

	Since the kernel $h(\cdot,\cdot)$ is bounded and $\sigma^2>0$, the
	Berry-Esseen theorem for nondegenerate $U$-statistics
	\cite{callaert1978berry} yields a constant $0<C_0<\infty$ such that
	\begin{align}\label{eq:BE_Mn}
		\sup_{t\in\mathbb{R}}
		\left|
		\Pr\left[
		\frac{\sqrt{n}(\mathsf{V}_n-\mu)}{\sigma}
		>t
		\right]
		-
		Q(t)
		\right|
		\le
		\frac{C_0}{\sqrt{n}}.
	\end{align}
	
	By \eqref{eq:D_Mn_bound}, for every $t\in\mathbb{R}$,
	\begin{align}\label{eq:neq}
		&
		\Pr\left[
		\frac{\sqrt{n}(\mathsf{V}_n-\mu)}{\sigma}
		>
		t+\frac{1}{\sigma\sqrt{n}}
		\right]
		\le
		\Pr\left[
		\frac{
			\sqrt{n}\bigl(
			\mathsf{D}(\underline{\mathsf{X}},
			\underline{\mathsf{Y}})
			-\mu
			\bigr)
		}{\sigma}
		>t
		\right]
		\le
		\Pr\left[
		\frac{\sqrt{n}(\mathsf{V}_n-\mu)}{\sigma}
		>
		t-\frac{1}{\sigma\sqrt{n}}
		\right].
	\end{align}
	Moreover, since
	\begin{align}
		|Q(t+u)-Q(t)|
		\le
		\frac{|u|}{\sqrt{2\pi}},
	\end{align}
	uniformly in $t,u\in\mathbb{R}$, \eqref{eq:BE_Mn} and \eqref{eq:neq} imply
	\begin{align}
		\sup_{t\in\mathbb{R}}
		\left|
		\Pr\left[
		\frac{
			\sqrt{n}\bigl(
			\mathsf{D}(\underline{\mathsf{X}},
			\underline{\mathsf{Y}})
			-\mu
			\bigr)
		}{\sigma}
		>t
		\right]
		-
		Q(t)
		\right|
		\le
		\frac{C}{\sqrt{n}},
	\end{align}
	where $C= C_0+1/(\sigma\sqrt{2\pi})$. 
	
	It remains to justify the non-strict inequality in the theorem
	statement. For every $t\in\mathbb R$, continuity of probability from above gives
	\begin{align}
		&
		\Pr\left[
		\frac{
			\sqrt n
			\left(
			\mathsf D(
			\underline{\mathsf X},
			\underline{\mathsf Y})
			-\mu
			\right)
		}{\sigma}
		\ge t
		\right]
		=
		\lim_{s\uparrow t}
		\Pr\left[
		\frac{
			\sqrt n
			\left(
			\mathsf D(
			\underline{\mathsf X},
			\underline{\mathsf Y})
			-\mu
			\right)
		}{\sigma}
		>s
		\right].
		\label{eq:nonstrict-tail-left-limit}
	\end{align}
	Applying the preceding uniform bound at $s<t$ and letting
	$s\uparrow t$, together with the continuity of $Q(\cdot)$, yields
	\begin{align}
		\left|
		\Pr\left[
		\frac{
			\sqrt n
			\left(
			\mathsf D(
			\underline{\mathsf X},
			\underline{\mathsf Y})
			-\mu
			\right)
		}{\sigma}
		\ge t
		\right]
		-
		Q(t)
		\right|
		\le
		\frac{C}{\sqrt n}.
	\end{align}
	Taking the supremum over $t\in\mathbb R$ proves the non-strict
	inequality and completes the proof.

	\subsection{Proof of Lemma \ref{lem:finite-saddlepoint}}\label{proof saddlepoint expansion}
	We first establish the existence and uniqueness of the
	finite-blocklength saddlepoint. For this purpose, define
	\begin{align}
		g_\theta(x)
		\triangleq
		x\frac{e^{\theta x}}{1+e^{\theta x}},
		\qquad
		(\theta,x)\in\mathbb R\times[0,1].
	\end{align}
	The derivatives of the finite- and limiting-blocklength CGFs can
	then be represented as
	\begin{align}
		K_n'(\theta)
		&=
		\frac{1}{n}
		\sum_{i=1}^{n}
		g_\theta\left(\frac{i}{n}\right),\\
		K'(\theta)
		&=
		\int_0^1g_\theta(x)\,\mathrm{d}x.
	\end{align}
	
	Since $\{\theta_d:d\in\mathcal D\}$ is a compact subset of
	$(-\infty,0)$, we may choose a compact interval
	$[a,b]\subset(-\infty,0)$ such that $K'(a)<\omega_1$ and $K'(b)>\omega_2$. The partial derivative $\partial g_\theta(x)/\partial x$ is bounded uniformly over $(\theta,x)\in[a,b]\times[0,1]$. Consequently, the right-endpoint Riemann-sum error satisfies
	\begin{align}
		\sup_{\theta\in[a,b]}
		\left|
		K_n'(\theta)-K'(\theta)
		\right|
		=
		O(n^{-1}).
		\label{eq:Kn-prime-uniform-convergence}
	\end{align}
	Together with \eqref{eq:dn-rounding}, this uniform convergence
	implies that, for all sufficiently large $n$,
	\begin{align}
		K_n'(a)
		<
		d_n
		<
		K_n'(b),
		\qquad
		d\in\mathcal D.
		\label{eq:finite-saddlepoint-bracket}
	\end{align}
	Because $K_n''(\theta)>0$, the function $K_n'(\cdot)$ is continuous and strictly increasing. Therefore, \eqref{eq:finite-blocklength-saddlepoint} has a unique solution
	$\theta_{n,d}\in(a,b)$.
	
	We next derive the uniform saddlepoint expansions. On the fixed
	compact interval $[a,b]$, the Euler-Maclaurin formula \cite{olver2010nist} gives
	\begin{align}
		K_n^{(r)}(\theta)
		=
		K^{(r)}(\theta)
		+
		\frac{G^{(r)}(\theta)}{n}
		+
		O(n^{-2}),
		\qquad
		r=0,1,2,
		\label{eq:uniform-Euler-Maclaurin-expansion}
	\end{align}
	where the remainder is uniform over $\theta\in[a,b]$. In
	particular, the cases $r=0,1,2$ respectively give
	\begin{align}
		K_n(\theta)
		&=
		K(\theta)
		+
		\frac{G(\theta)}{n}
		+
		O(n^{-2}),
		\label{eq:Kn-Euler-Maclaurin}\\
		K_n'(\theta)
		&=
		K'(\theta)
		+
		\frac{G'(\theta)}{n}
		+
		O(n^{-2}),
		\label{eq:Kn-prime-Euler-Maclaurin}\\
		K_n''(\theta)
		&=
		K''(\theta)
		+
		\frac{G''(\theta)}{n}
		+
		O(n^{-2}).
		\label{eq:Kn-second-Euler-Maclaurin}
	\end{align}
	
	Combining
	$K_n'(\theta_{n,d})=d_n$,
	$K'(\theta_d)=d$, and
	\eqref{eq:Kn-prime-Euler-Maclaurin}, we obtain
	\begin{align}
		K'(\theta_{n,d})-K'(\theta_d)
		=
		d_n-d
		-
		\frac{G'(\theta_{n,d})}{n}
		+
		O(n^{-2})
		=
		O(n^{-1}),
		\label{eq:limiting-derivative-difference}
	\end{align}
	uniformly over $d\in\mathcal D$. Since $K''(\cdot)$ is bounded away from zero on $[a,b]$, the mean-value theorem applied to $K'(\cdot)$ yields
	\begin{align}
		\theta_{n,d}-\theta_d
		=
		O(n^{-1}),
	\end{align}
	which proves \eqref{eq:theta-nd-expansion}.
	
	The expansion of the second derivative follows from
	\eqref{eq:Kn-second-Euler-Maclaurin} and
	\eqref{eq:theta-nd-expansion}. More precisely, the smoothness of
	$K''(\cdot)$ on $[a,b]$ gives
	\begin{align}
		K_n''(\theta_{n,d})
		&=
		K''(\theta_{n,d})+O(n^{-1})=
		K''(\theta_d)+O(n^{-1}),
	\end{align}
	which proves \eqref{eq:Kn-second-expansion}.
	
	It remains to expand the finite-blocklength rate function. From
	\eqref{eq:Kn-Euler-Maclaurin}, we have
	\begin{align}
		I_n(d_n)
		&=
		\theta_{n,d}d_n
		-
		K(\theta_{n,d})
		-
		\frac{G(\theta_{n,d})}{n}
		+
		O(n^{-2}).
		\label{eq:finite-rate-function-intermediate}
	\end{align}
	The function $L_d(\theta)\triangleq\theta d-K(\theta)$
	satisfies $L_d'(\theta_d)=0$. A Taylor expansion around
	$\theta_d$, together with
	\eqref{eq:theta-nd-expansion}, therefore gives
	\begin{align}
		\theta_{n,d}d-K(\theta_{n,d})
		=
		I(d)+O(n^{-2}).
		\label{eq:stationary-rate-expansion}
	\end{align}
	In addition, \eqref{eq:dn-rounding} and
	\eqref{eq:theta-nd-expansion} imply
	\begin{align}
		\theta_{n,d}(d_n-d)
		=
		O(n^{-2}),
		\qquad
		G(\theta_{n,d})
		=
		G(\theta_d)+O(n^{-1}).
	\end{align}
	Substituting these estimates into
	\eqref{eq:finite-rate-function-intermediate} yields
	\begin{align}
		I_n(d_n)
		=
		I(d)
		-
		\frac{G(\theta_d)}{n}
		+
		O(n^{-2}).
	\end{align}
	Multiplying both sides by $n$ proves
	\eqref{eq:finite-rate-function-expansion}. All the estimates used
	above are uniform over $d\in\mathcal D$. This completes the proof.

	\subsection{Proof of Lemma~\ref{lem:uniform-lattice-lclt}}
	\label{app:uniform-lattice-lclt}
	
	Under the tilted measure $\Pr\nolimits_{n,d}$, the random variables
	$\{\mathsf B_i\}_{i=1}^n$ remain independent, with
	\begin{align}
		p_{i,n,d}
		\triangleq
		\Pr\nolimits_{n,d}[\mathsf B_i=1]
		=
		\frac{
			e^{\theta_{n,d}i/n}
		}{
			1+e^{\theta_{n,d}i/n}
		},
		\qquad
		i\in\{1,\ldots,n\}.
		\label{eq:app-tilted-parameter}
	\end{align}
	To center the weighted Bernoulli sum under $\Pr\nolimits_{n,d}$, define
	\begin{align}
		\mathsf A_{i,n,d}
		\triangleq
		i\bigl(\mathsf B_i-p_{i,n,d}\bigr),
		\qquad
		i\in\{1,\ldots,n\}.
		\label{eq:app-centered-summands}
	\end{align}
	The exact centering relation in \eqref{eq:tilted-mean} then gives
	\begin{align}
		\zeta_n-\kappa_n(d)
		=
		\sum_{i=1}^n
		\mathsf A_{i,n,d}.
		\label{eq:app-centered-sum}
	\end{align}
	Moreover, the variance of the centered sum satisfies
	\begin{align}
		\sum_{i=1}^n
		\mathbb E_{n,d}
		\left[
		\mathsf A_{i,n,d}^2
		\right]
		=
		\sigma_{n,d}^2.
		\label{eq:app-centered-variance}
	\end{align}
	
	We first establish uniform bounds on the tilted Bernoulli parameters. By Lemma~\ref{lem:finite-saddlepoint}, there exists a compact interval $[a,b]\subset(-\infty,0)$ such that
	$\theta_{n,d}\in[a,b]$ for all sufficiently large $n$ and every
	$d\in\mathcal D$. Since $i/n\in(0,1]$, we have
	\begin{align}
		a
		\le
		\frac{\theta_{n,d}i}{n}
		\le
		0.
	\end{align}
	It follows that there exists a constant $p_\star\in(0,1/2)$,
	independent of $n$, $d$, and $i$, such that
	\begin{align}
		p_\star
		\le
		p_{i,n,d}
		\le
		\frac12.
		\label{eq:app-tilted-parameter-bounds}
	\end{align}
	Consequently, the tilted Bernoulli variances satisfy
	\begin{align}
		p_\star(1-p_\star)
		\le
		p_{i,n,d}(1-p_{i,n,d})
		\le
		\frac14.
		\label{eq:app-Bernoulli-variance-bounds}
	\end{align}
	
	We next establish a Lyapunov-type moment bound uniformly over the
	tilted family \cite{petrov2012sums}.
	Since
	$|\mathsf B_i-p_{i,n,d}|\le 1$, the third absolute moment of each
	centered summand satisfies
	\begin{align}
		\mathbb E_{n,d}
		\left[
		|\mathsf A_{i,n,d}|^3
		\right]
		\le
		i^3.
		\label{eq:app-third-moment-bound}
	\end{align}
	Summing these bounds over $i\in\{1,\ldots,n\}$ gives
	\begin{align}
		\sum_{i=1}^n
		\mathbb E_{n,d}
		\left[
		|\mathsf A_{i,n,d}|^3
		\right]
		\le
		\sum_{i=1}^n i^3
		=
		O(n^4).
		\label{eq:app-total-third-moment-bound}
	\end{align}
	Combining \eqref{eq:app-total-third-moment-bound} with the uniform
	variance bound in \eqref{eq:tilted-variance-order}, we obtain
	\begin{align}
		\sup_{d\in\mathcal D}
		\frac{
			\displaystyle
			\sum_{i=1}^n
			\mathbb E_{n,d}
			\left[
			|\mathsf A_{i,n,d}|^3
			\right]
		}{
			\sigma_{n,d}^3
		}
		=
		O(n^{-1/2}).
		\label{eq:app-uniform-Lyapunov-bound}
	\end{align}
	Thus, the standardized triangular array satisfies the Lyapunov
	condition uniformly over $d\in\mathcal D$.
	
	We now use the preceding moment estimate to analyze the characteristic function of the standardized centered sum. For each
	$d\in\mathcal D$, define
	\begin{align}
		\psi_{n,d}(t)
		\triangleq
		\mathbb E_{n,d}
		\left[
		\exp\left(
		\mathrm{j}t
		\frac{
			\zeta_n-\kappa_n(d)
		}{
			\sigma_{n,d}
		}
		\right)
		\right],
		\qquad
		t\in\mathbb R,
		\label{eq:app-standardized-characteristic-function}
	\end{align}
	where $\mathrm{j}\triangleq\sqrt{-1}$ denotes the imaginary unit.
	Using \eqref{eq:app-centered-sum} and the independence of
	$\{\mathsf A_{i,n,d}\}_{i=1}^n$ under $\Pr\nolimits_{n,d}$, this characteristic function factorizes as
	\begin{align}
		\psi_{n,d}(t)
		=
		\prod_{i=1}^n
		\mathbb E_{n,d}
		\left[
		\exp\left(
		\frac{\mathrm{j}t\mathsf A_{i,n,d}}
		{\sigma_{n,d}}
		\right)
		\right].
		\label{eq:app-characteristic-function-product}
	\end{align}
	
	To compare each characteristic function factor with its Gaussian
	counterpart, define
	\begin{align}
		\alpha_{i,n,d}(t)
		&\triangleq
		\mathbb E_{n,d}
		\left[
		\exp\left(
		\frac{\mathrm{j}t\mathsf A_{i,n,d}}
		{\sigma_{n,d}}
		\right)
		\right],
		\label{eq:app-alpha-factor}\\
		\beta_{i,n,d}(t)
		&\triangleq
		\exp\left(
		-
		\frac{
			t^2
			\mathbb E_{n,d}
			[
			\mathsf A_{i,n,d}^2
			]
		}{
			2\sigma_{n,d}^2
		}
		\right).
		\label{eq:app-beta-factor}
	\end{align}
	The product representation in
	\eqref{eq:app-characteristic-function-product} can therefore be
	rewritten as
	\begin{align}
		\psi_{n,d}(t)
		=
		\prod_{i=1}^n
		\alpha_{i,n,d}(t).
		\label{eq:app-psi-alpha-product}
	\end{align}
	
	We first compare both factors with the same second-order
	approximation. The elementary Taylor bound
	\begin{align}
		\left|
		e^{\mathrm{j}x}
		-
		1
		-
		\mathrm{j}x
		+
		\frac{x^2}{2}
		\right|
		\le
		\frac{|x|^3}{6},
		\qquad
		x\in\mathbb R,
		\label{eq:app-complex-exponential-Taylor-bound}
	\end{align}
	together with
	$\mathbb E_{n,d}[\mathsf A_{i,n,d}]=0$, implies
	\begin{align}
		\Bigg|
		\alpha_{i,n,d}(t)
		-
		\left(
		1
		-
		\frac{
			t^2
			\mathbb E_{n,d}
			[
			\mathsf A_{i,n,d}^2
			]
		}{
			2\sigma_{n,d}^2
		}
		\right)
		\Bigg|
		\le
		\frac{
			|t|^3
			\mathbb E_{n,d}
			[
			|\mathsf A_{i,n,d}|^3
			]
		}{
			6\sigma_{n,d}^3
		}.
		\label{eq:app-individual-characteristic-Taylor-bound}
	\end{align}
	
	The Gaussian factor admits an analogous second-order approximation.
	For every $x\ge0$, the real exponential satisfies
	\begin{align}
		\left|
		e^{-x}-(1-x)
		\right|
		\le
		\frac{x^2}{2}.
		\label{eq:app-real-exponential-Taylor-bound}
	\end{align}
	Applying this inequality with $x=\frac{t^2\mathbb E_{n,d}[\mathsf A_{i,n,d}^2]}{2\sigma_{n,d}^2}$, we have
	\begin{align}
		\Bigg|
		\beta_{i,n,d}(t)
		-
		\left(
		1
		-
		\frac{
			t^2
			\mathbb E_{n,d}
			[
			\mathsf A_{i,n,d}^2
			]
		}{
			2\sigma_{n,d}^2
		}
		\right)
		\Bigg|
		\le
		\frac{
			t^4
			\left(
			\mathbb E_{n,d}
			[
			\mathsf A_{i,n,d}^2
			]
			\right)^2
		}{
			8\sigma_{n,d}^4
		}.
		\label{eq:app-individual-Gaussian-factor-bound}
	\end{align}
	
	Since both $\alpha_{i,n,d}(t)$ and $\beta_{i,n,d}(t)$ are compared
	with the same quadratic approximation, the triangle inequality yields
	\begin{align}
		\left|
		\alpha_{i,n,d}(t)
		-
		\beta_{i,n,d}(t)
		\right|
		&\le
		\frac{
			|t|^3
			\mathbb E_{n,d}
			[
			|\mathsf A_{i,n,d}|^3
			]
		}{
			6\sigma_{n,d}^3
		}
		+
		\frac{
			t^4
			\left(
			\mathbb E_{n,d}
			[
			\mathsf A_{i,n,d}^2
			]
			\right)^2
		}{
			8\sigma_{n,d}^4
		}.
		\label{eq:app-individual-factor-comparison}
	\end{align}
	
	We next extend the single-factor comparison to the full product.
	The difference between the two products admits the following telescoping decomposition
	\begin{align}
		\prod_{i=1}^n
		\alpha_{i,n,d}(t)
		-
		\prod_{i=1}^n
		\beta_{i,n,d}(t)
		&=
		\sum_{r=1}^n
		\left(
		\prod_{i=1}^{r-1}
		\alpha_{i,n,d}(t)
		\right)
		\left(
		\alpha_{r,n,d}(t)
		-
		\beta_{r,n,d}(t)
		\right)
		\times
		\left(
		\prod_{i=r+1}^n
		\beta_{i,n,d}(t)
		\right).
		\label{eq:app-telescoping-product-identity}
	\end{align}
	Each $\alpha_{i,n,d}(t)$ is a characteristic function factor and
	therefore satisfies $\left|\alpha_{i,n,d}(t)\right|\le1$. The definition of $\beta_{i,n,d}(t)$ similarly ensures that $0<\beta_{i,n,d}(t)\le1$. Taking absolute values in
	\eqref{eq:app-telescoping-product-identity} and applying these two
	bounds gives
	\begin{align}
		\left|
		\prod_{i=1}^n
		\alpha_{i,n,d}(t)
		-
		\prod_{i=1}^n
		\beta_{i,n,d}(t)
		\right|
		\le
		\sum_{i=1}^n
		\left|
		\alpha_{i,n,d}(t)
		-
		\beta_{i,n,d}(t)
		\right|.
		\label{eq:app-telescoping-product-bound}
	\end{align}
	
	The product of the Gaussian factors can be evaluated exactly. Using
	the variance identity in \eqref{eq:app-centered-variance}, we obtain
	\begin{align}
		\prod_{i=1}^n
		\beta_{i,n,d}(t)
		&=
		\exp\left(
		-
		\frac{t^2}{2\sigma_{n,d}^2}
		\sum_{i=1}^n
		\mathbb E_{n,d}
		\left[
		\mathsf A_{i,n,d}^2
		\right]
		\right)
		=
		e^{-t^2/2}.
		\label{eq:app-Gaussian-factor-product}
	\end{align}
	Combining
	\eqref{eq:app-psi-alpha-product}
	\eqref{eq:app-individual-factor-comparison},
	\eqref{eq:app-telescoping-product-bound}, and
	\eqref{eq:app-Gaussian-factor-product}, we arrive at
	\begin{align}
		\left|
		\psi_{n,d}(t)
		-
		e^{-t^2/2}
		\right|
		&\le
		\frac{|t|^3}{6\sigma_{n,d}^3}
		\sum_{i=1}^n
		\mathbb E_{n,d}
		\left[
		|\mathsf A_{i,n,d}|^3
		\right]
		+
		\frac{t^4}{8\sigma_{n,d}^4}
		\sum_{i=1}^n
		\left(
		\mathbb E_{n,d}
		\left[
		\mathsf A_{i,n,d}^2
		\right]
		\right)^2.
		\label{eq:app-characteristic-Gaussian-comparison}
	\end{align}
	
	It remains to control the second sum in
	\eqref{eq:app-characteristic-Gaussian-comparison}. Since
	$|\mathsf A_{i,n,d}|\le i$, we can obtain
	\begin{align}
		\sum_{i=1}^n
		\left(
		\mathbb E_{n,d}
		\left[
		\mathsf A_{i,n,d}^2
		\right]
		\right)^2
		\le
		\sum_{i=1}^n i^4
		=
		O(n^5).
		\label{eq:app-squared-variance-sum-bound}
	\end{align}
	Therefore, for every fixed $\tau_0>0$, there exists a constant
	$0<C_{\tau_0}<\infty$, independent of $n$ and $d$, such that, for all sufficiently large $n$,
	\begin{align}
		\sup_{d\in\mathcal D}
		\sup_{|t|\le \tau_0}
		\left|
		\psi_{n,d}(t)
		-
		e^{-t^2/2}
		\right|
		\le
		\frac{C_{\tau_0}}{\sqrt n}.
		\label{eq:app-compact-characteristic-convergence}
	\end{align}
	In particular, the left-hand side of
	\eqref{eq:app-compact-characteristic-convergence} converges to zero
	as $n\to\infty$.
	
	We next control the characteristic function outside bounded frequency
	intervals. For each individual factor defined in
	\eqref{eq:app-alpha-factor}, we have
	\begin{align}
		\left|
		\alpha_{i,n,d}(t)
		\right|
		=
		\left|
		1-p_{i,n,d}
		+
		p_{i,n,d}
		e^{\mathrm{j}it/\sigma_{n,d}}
		\right|.
		\label{eq:app-alpha-modulus-representation}
	\end{align}
	A direct calculation of the squared modulus gives
	\begin{align}
		\left|
		\alpha_{i,n,d}(t)
		\right|^2
		=
		1
		-
		4p_{i,n,d}(1-p_{i,n,d})
		\sin^2\left(
		\frac{it}{2\sigma_{n,d}}
		\right).
		\label{eq:app-alpha-squared-modulus}
	\end{align}
	Using the inequality
	$\sqrt{1-x}\le e^{-x/2}$ for $x\in[0,1]$, together with
	\eqref{eq:app-Bernoulli-variance-bounds}, we obtain
	\begin{align}
		\left|
		\alpha_{i,n,d}(t)
		\right|
		\le
		\exp\left[
		-2p_\star(1-p_\star)
		\sin^2\left(
		\frac{it}{2\sigma_{n,d}}
		\right)
		\right].
		\label{eq:app-individual-characteristic-decay}
	\end{align}
	Multiplying the individual bounds in
	\eqref{eq:app-individual-characteristic-decay} yields
	\begin{align}
		\left|
		\psi_{n,d}(t)
		\right|
		\le
		\exp\left[
		-2p_\star(1-p_\star)
		\sum_{i=1}^n
		\sin^2\left(
		\frac{it}{2\sigma_{n,d}}
		\right)
		\right].
		\label{eq:app-characteristic-decay-preliminary}
	\end{align}
	
	We next establish an auxiliary lower bound for the trigonometric sum
	in \eqref{eq:app-characteristic-decay-preliminary}. Specifically, there exists a constant $0<C_2<\infty$ such that, for all sufficiently large $n$ and every $x\in[-\pi,\pi]$,
	\begin{align}
		\sum_{i=1}^n
		\sin^2\left(
		\frac{ix}{2}
		\right)
		\ge
		C_2n
		\min\left\{
		1,n^2x^2
		\right\}.
		\label{eq:app-trigonometric-sum-lower-bound}
	\end{align}
	To prove \eqref{eq:app-trigonometric-sum-lower-bound}, observe first
	that both sides are even functions of $x$. It therefore suffices to
	consider $x\in[0,\pi]$. Fix a constant $\xi_0>2\pi$, and divide the
	proof into the following three regimes:
	\begin{align}
		nx\le1,\qquad
		1\le nx\le\xi_0,\qquad
		\xi_0\le nx\le n\pi.
	\end{align}
	
	\begin{itemize}
		\item In the first regime, where $nx\le1$, we have
		$ix/2\le1/2$ for every $i\in\{1,\ldots,n\}$. The inequality
		$\sin y\ge 2y/\pi$ for $y\in[0,\pi/2]$ therefore gives
		\begin{align}
			\sum_{i=1}^n
			\sin^2\left(
			\frac{ix}{2}
			\right)
			&\ge
			\frac{x^2}{\pi^2}
			\sum_{i=1}^n i^2
			\ge
			\frac{1}{3\pi^2}n^3x^2
			=
			\frac{1}{3\pi^2}n(n^2x^2).
			\label{eq:app-trigonometric-small-frequency}
		\end{align}
		\item In the second regime, the parameter $nx$ ranges over the compact interval $[1,\xi_0]$. The corresponding Riemann sums satisfy
		\begin{align}
			\sup_{1\le nx\le\xi_0}
			\left|
			\frac{1}{n}
			\sum_{i=1}^n
			\sin^2\left(
			\frac{ix}{2}
			\right)
			-
			\int_0^1
			\sin^2\left(
			\frac{nxy}{2}
			\right)
			\,\mathrm{d}y
			\right|
			\longrightarrow0.
			\label{eq:app-uniform-trigonometric-Riemann-sum}
		\end{align}
		Moreover, the limiting integrals are uniformly bounded away from
		zero, since
		\begin{align}
			\inf_{1\le v\le\xi_0}
			\int_0^1
			\sin^2\left(
			\frac{vy}{2}
			\right)
			\,\mathrm{d}y
			=
			\inf_{1\le v\le\xi_0}
			\left(
			\frac12
			-
			\frac{\sin v}{2v}
			\right)
			>0.
			\label{eq:app-trigonometric-integral-lower-bound}
		\end{align}

		Consequently, there exists a constant $0<C_3<\infty$ such that, for all
		sufficiently large $n$,
		\begin{align}
			\sum_{i=1}^n
			\sin^2\left(
			\frac{ix}{2}
			\right)
			\ge
			C_3n,
			\qquad
			1\le nx\le\xi_0.
			\label{eq:app-trigonometric-intermediate-frequency}
		\end{align}
		\item In the third regime, where
		$\xi_0\le nx\le n\pi$, the finite trigonometric-sum identity
		\begin{align}
			\sum_{i=1}^n
			\sin^2\left(
			\frac{ix}{2}
			\right)
			=
			\frac{n}{2}
			-
			\frac{
				\sin(nx/2)
				\cos((n+1)x/2)
			}{
				2\sin(x/2)
			}
			\label{eq:app-trigonometric-sum-identity}
		\end{align}
		provides a direct lower bound. Since $x\in(0,\pi]$, the elementary inequality $\sin\left(\frac{x}{2}\right)\ge\frac{x}{\pi}$ holds. Combining this inequality with
		\eqref{eq:app-trigonometric-sum-identity} gives
		\begin{align}
			\sum_{i=1}^n
			\sin^2\left(
			\frac{ix}{2}
			\right)
			&\ge
			\frac{n}{2}
			-
			\frac{1}{
				2\sin(x/2)
			}
			\ge
			\frac{n}{2}
			-
			\frac{\pi}{2x}
			\ge
			n\left(
			\frac12
			-
			\frac{\pi}{2\xi_0}
			\right),
			\label{eq:app-trigonometric-large-frequency}
		\end{align}
		where the last inequality follows from $nx\ge\xi_0$. Since
		$\xi_0>2\pi$, the coefficient on the right-hand side is strictly
		positive.
	\end{itemize}
	Combining the bounds obtained in the three regimes yields a
	constant $0<C_2<\infty$ for which \eqref{eq:app-trigonometric-sum-lower-bound} holds.
	
	We now apply this auxiliary bound to the standardized characteristic
	function. For every $t$ satisfying
	$|t|\le\pi\sigma_{n,d}$, the choice $x=\frac{t}{\sigma_{n,d}}$ ensures that $x\in[-\pi,\pi]$. Therefore,
	\eqref{eq:app-trigonometric-sum-lower-bound} gives
	\begin{align}
		\sum_{i=1}^n
		\sin^2\left(
		\frac{it}{2\sigma_{n,d}}
		\right)
		\ge
		C_2n
		\min\left\{
		1,
		\frac{n^2t^2}{\sigma_{n,d}^2}
		\right\}.
		\label{eq:app-standardized-trigonometric-bound}
	\end{align}
	The uniform upper bound on the tilted variance in
	\eqref{eq:tilted-variance-order} implies
	\begin{align}
		\frac{n^2t^2}{\sigma_{n,d}^2}
		\ge
		\frac{t^2}{C_1n}.
	\end{align}
	Consequently, we obtain
	\begin{align}
		n
		\min\left\{
		1,
		\frac{n^2t^2}{\sigma_{n,d}^2}
		\right\}
		&\ge
		n
		\min\left\{
		1,
		\frac{t^2}{C_1n}
		\right\}
		=
		\min\left\{
		n,
		\frac{t^2}{C_1}
		\right\}
		\ge
		\min\left\{
		1,C_1^{-1}
		\right\}
		\min\left\{
		t^2,n
		\right\}.
		\label{eq:app-standardized-frequency-comparison}
	\end{align}

	Substituting
	\eqref{eq:app-standardized-trigonometric-bound} and
	\eqref{eq:app-standardized-frequency-comparison} into
	\eqref{eq:app-characteristic-decay-preliminary}, we obtain
	\begin{align}
		\left|
		\psi_{n,d}(t)
		\right|
		\le
		\exp\left(
		-C_4
		\min\left\{
		t^2,n
		\right\}
		\right),
		\qquad
		|t|\le\pi\sigma_{n,d},
		\label{eq:app-uniform-characteristic-decay}
	\end{align}
	where $0<C_4<\infty$ is independent of $n$, $d$, and $t$.
	
	The compact-frequency approximation in
	\eqref{eq:app-compact-characteristic-convergence} and the uniform
	decay estimate in \eqref{eq:app-uniform-characteristic-decay} provide the two ingredients required for lattice Fourier inversion.
	
	For every $k\in\mathbb Z$, define the corresponding standardized
	lattice point as
	\begin{align}
		v_{n,d}(k)
		\triangleq
		\frac{k-\kappa_n(d)}
		{\sigma_{n,d}}.
		\label{eq:app-standardized-lattice-point}
	\end{align}
	Since $\zeta_n$ is integer-valued with lattice span one, the
	lattice Fourier inversion formula \cite{petrov2012sums} gives
	\begin{align}
		\sigma_{n,d}
		\Pr\nolimits_{n,d}[\zeta_n=k]
		=
		\frac{1}{2\pi}
		\int_{-\pi\sigma_{n,d}}^{\pi\sigma_{n,d}}
		e^{-\mathrm jt v_{n,d}(k)}
		\psi_{n,d}(t)
		\,\mathrm{d}t.
		\label{eq:app-lattice-Fourier-inversion}
	\end{align}
	The standard Gaussian density admits the Fourier representation
	\begin{align}
		\phi\bigl(v_{n,d}(k)\bigr)
		=
		\frac{1}{2\pi}
		\int_{\mathbb R}
		e^{-\mathrm jt v_{n,d}(k)}
		e^{-t^2/2}
		\,\mathrm{d}t.
		\label{eq:app-Gaussian-Fourier-representation}
	\end{align}
	
	The lower bound in \eqref{eq:tilted-variance-order} ensures that
	$\sigma_{n,d}\to\infty$ uniformly over $d\in\mathcal D$. Hence, for
	every fixed $\tau_0>0$, the interval $[-\tau_0,\tau_0]$ is contained
	in $[-\pi\sigma_{n,d},\pi\sigma_{n,d}]$ for all sufficiently large
	$n$ and every $d\in\mathcal D$.
	
	Subtracting \eqref{eq:app-Gaussian-Fourier-representation} from
	\eqref{eq:app-lattice-Fourier-inversion}, taking absolute values, and
	using
	$\left|e^{-\mathrm{j}t v_{n,d}(k)}\right|=1$, we obtain
	\begin{align}
		&
		\left|
		\sigma_{n,d}
		\Pr\nolimits_{n,d}[\zeta_n=k]
		-
		\phi\bigl(v_{n,d}(k)\bigr)
		\right|
		\notag\\
		&\le
		\frac{1}{2\pi}
		\int_{|t|\le\tau_0}
		\left|
		\psi_{n,d}(t)-e^{-t^2/2}
		\right|
		\,\mathrm{d}t
		+
		\frac{1}{2\pi}
		\int_{\tau_0<|t|\le\pi\sigma_{n,d}}
		\left|
		\psi_{n,d}(t)
		\right|
		\,\mathrm{d}t
		+
		\frac{1}{2\pi}
		\int_{|t|>\tau_0}
		e^{-t^2/2}
		\,\mathrm{d}t.
		\label{eq:app-Fourier-error-decomposition}
	\end{align}
	
	The compact-frequency estimate in
	\eqref{eq:app-compact-characteristic-convergence} gives
	\begin{align}
		\sup_{d\in\mathcal D}
		\int_{|t|\le\tau_0}
		\left|
		\psi_{n,d}(t)-e^{-t^2/2}
		\right|
		\,\mathrm{d}t
		\le
		\frac{2\tau_0C_{\tau_0}}{\sqrt n},
		\label{eq:app-compact-Fourier-error}
	\end{align}
	which converges to zero as $n\to\infty$ for every fixed
	$\tau_0>0$.
	
	The characteristic function tail is controlled by
	\eqref{eq:app-uniform-characteristic-decay}. For all sufficiently
	large $n$, splitting the integration interval at $\sqrt n$ gives
	\begin{align}
		&
		\sup_{d\in\mathcal D}
		\int_{\tau_0<|t|\le\pi\sigma_{n,d}}
		\left|
		\psi_{n,d}(t)
		\right|
		\,\mathrm{d}t
		\le
		2\int_{\tau_0}^{\sqrt n}
		e^{-C_4t^2}
		\,\mathrm{d}t
		+
		2\pi
		\sup_{d\in\mathcal D}
		\sigma_{n,d}
		e^{-C_4n}.
		\label{eq:app-characteristic-tail-integral}
	\end{align}
	The upper variance bound in
	\eqref{eq:tilted-variance-order} implies $\sigma_{n,d}=O(n^{3/2})$, so we obtain
	\begin{align}
		&
		\sup_{d\in\mathcal D}
		\int_{\tau_0<|t|\le\pi\sigma_{n,d}}
		\left|
		\psi_{n,d}(t)
		\right|
		\,\mathrm{d}t
		\le
		2\int_{\tau_0}^{\infty}
		e^{-C_4t^2}
		\,\mathrm{d}t
		+
		O\left(
		n^{3/2}e^{-C_4n}
		\right).
		\label{eq:app-characteristic-tail-final-bound}
	\end{align}
	
	The last term in
	\eqref{eq:app-Fourier-error-decomposition} is independent of
	$d$ and $k$, and satisfies
	\begin{align}
		\int_{|t|>\tau_0}
		e^{-t^2/2}
		\,\mathrm{d}t
		=
		2\int_{\tau_0}^{\infty}
		e^{-t^2/2}
		\,\mathrm{d}t.
		\label{eq:app-Gaussian-tail-integral}
	\end{align}
	
	Combining
	\eqref{eq:app-Fourier-error-decomposition},
	\eqref{eq:app-compact-Fourier-error},
	\eqref{eq:app-characteristic-tail-final-bound}, and
	\eqref{eq:app-Gaussian-tail-integral}, and then taking the upper limit as $n\to\infty$, we obtain
	\begin{align}
		&
		\limsup_{n\to\infty}
		\sup_{d\in\mathcal D}
		\sup_{k\in\mathbb Z}
		\left|
		\sigma_{n,d}
		\Pr\nolimits_{n,d}[\zeta_n=k]
		-
		\phi\left(
		\frac{k-\kappa_n(d)}
		{\sigma_{n,d}}
		\right)
		\right|
		\le
		\frac{1}{\pi}
		\int_{\tau_0}^{\infty}
		e^{-C_4t^2}
		\,\mathrm{d}t
		+
		\frac{1}{\pi}
		\int_{\tau_0}^{\infty}
		e^{-t^2/2}
		\,\mathrm{d}t.
		\label{eq:app-local-limit-limsup-bound}
	\end{align}
	Both integrals on the right-hand side converge to zero as
	$\tau_0\to\infty$. Since $\tau_0>0$ is arbitrary, we conclude that
	\begin{align}
		\lim_{n\to\infty}
		\sup_{d\in\mathcal D}
		\sup_{k\in\mathbb Z}
		\left|
		\sigma_{n,d}
		\Pr\nolimits_{n,d}[\zeta_n=k]
		-
		\phi\left(
		\frac{k-\kappa_n(d)}
		{\sigma_{n,d}}
		\right)
		\right|
		=
		0.
		\label{eq:app-uniform-local-Gaussian-conclusion}
	\end{align}
	This completes the proof.
	
	\subsection{Proof of Lemma~\ref{lem:one-sided-lattice-contribution}}
	\label{app:one-sided-lattice-contribution}
	
	Define the uniform local approximation error as
	\begin{align}
		\varepsilon_n
		\triangleq
		\sup_{d\in\mathcal D}
		\sup_{k\in\mathbb Z}
		\left|
		\sigma_{n,d}
		\Pr\nolimits_{n,d}[\zeta_n=k]
		-
		\phi\left(
		\frac{k-\kappa_n(d)}
		{\sigma_{n,d}}
		\right)
		\right|.
		\label{eq:uniform-local-approximation-error}
	\end{align}
	Lemma~\ref{lem:uniform-lattice-lclt} ensures that
	$\varepsilon_n\to0$ as $n\to\infty$.
	
	Evaluating the approximation error at $k=\kappa_n(d)-s$ and using the symmetry of the Gaussian density give
	\begin{align}
		\left|
		\Pr\nolimits_{n,d}
		\left[
		\zeta_n=\kappa_n(d)-s
		\right]
		-
		\frac{1}{\sigma_{n,d}}
		\phi\left(
		\frac{s}{\sigma_{n,d}}
		\right)
		\right|
		\le
		\frac{\varepsilon_n}{\sigma_{n,d}},
		\label{eq:local-probability-error-bound}
	\end{align}
	uniformly over $d\in\mathcal D$ and $s\in\{0,\ldots,\kappa_n(d)\}$.
	Substituting this estimate into \eqref{eq:one-sided-lattice-term} yields
	\begin{align}
		&
		\left|
		T_{n,d}
		-
		\frac{1}{\sigma_{n,d}}
		\sum_{s=0}^{\kappa_n(d)}
		e^{\theta_{n,d}s/n}
		\phi\left(
		\frac{s}{\sigma_{n,d}}
		\right)
		\right|
		\le
		\frac{\varepsilon_n}{\sigma_{n,d}}
		\sum_{s=0}^{\kappa_n(d)}
		e^{\theta_{n,d}s/n}.
		\label{eq:lattice-contribution-local-error}
	\end{align}
	
	The finite-blocklength saddlepoints remain uniformly bounded away
	from zero on the negative half-line. In particular, there exists a constant $b<0$, independent of $n$ and $d$, such that
	$\theta_{n,d}\le b$ for all sufficiently large $n$ and every
	$d\in\mathcal D$. Consequently,
	\begin{align}
		\sum_{s=0}^{\kappa_n(d)}
		e^{\theta_{n,d}s/n}
		\le
		\sum_{s=0}^{\infty}
		e^{bs/n}
		=
		\frac{1}{1-e^{b/n}}
		=
		O(n).
		\label{eq:geometric-weight-upper-bound}
	\end{align}
	Combining this estimate with
	\eqref{eq:tilted-variance-order} and $\varepsilon_n\to0$ gives
	\begin{align}
		\sup_{d\in\mathcal D}
		\left|
		T_{n,d}
		-
		\frac{1}{\sigma_{n,d}}
		\sum_{s=0}^{\kappa_n(d)}
		e^{\theta_{n,d}s/n}
		\phi\left(
		\frac{s}{\sigma_{n,d}}
		\right)
		\right|
		=
		o(n^{-1/2}).
		\label{eq:lattice-contribution-local-error-order}
	\end{align}
	
	We next replace the Gaussian density in the weighted sum by its value at the origin. Since $\phi'(\cdot)$ is bounded on $\mathbb R$, there exists a constant $0<C<\infty$ such that
	\begin{align}
		\left|
		\phi\left(
		\frac{s}{\sigma_{n,d}}
		\right)
		-
		\phi(0)
		\right|
		\le
		C
		\frac{s}{\sigma_{n,d}}.
		\label{eq:Gaussian-density-origin-error}
	\end{align}
	The resulting replacement error is bounded by
	\begin{align}
		&
		\frac{1}{\sigma_{n,d}}
		\sum_{s=0}^{\kappa_n(d)}
		e^{\theta_{n,d}s/n}
		\left|
		\phi\left(
		\frac{s}{\sigma_{n,d}}
		\right)
		-
		\phi(0)
		\right|
		\le
		\frac{C}{\sigma_{n,d}^2}
		\sum_{s=0}^{\infty}
		s e^{bs/n}
		=
		O(n^{-1}),
		\label{eq:Gaussian-density-replacement-error}
	\end{align}
	uniformly over $d\in\mathcal D$.
	
	Combining
	\eqref{eq:lattice-contribution-local-error-order} and
	\eqref{eq:Gaussian-density-replacement-error}, we obtain
	\begin{align}
		T_{n,d}
		=
		\frac{\phi(0)}{\sigma_{n,d}}
		\sum_{s=0}^{\kappa_n(d)}
		e^{\theta_{n,d}s/n}
		+
		o(n^{-1/2}),
		\label{eq:lattice-contribution-geometric-reduction}
	\end{align}
	uniformly over $d\in\mathcal D$.
	
	The remaining finite geometric sum can be evaluated explicitly as
	\begin{align}
		\sum_{s=0}^{\kappa_n(d)}
		e^{\theta_{n,d}s/n}
		=
		\frac{
			1-
			e^{\theta_{n,d}(\kappa_n(d)+1)/n}
		}{
			1-e^{\theta_{n,d}/n}
		}.
		\label{eq:finite-geometric-sum}
	\end{align}
	Since $d\ge\omega_1>0$ and $\theta_{n,d}$ remains uniformly bounded
	away from zero on the negative half-line, the exponential term in the numerator is $O(e^{-Cn})$ for some $C>0$, uniformly over
	$d\in\mathcal D$. The denominator satisfies
	\begin{align}
		1-e^{\theta_{n,d}/n}
		=
		-\frac{\theta_{n,d}}{n}
		\bigl(1+O(n^{-1})\bigr),
	\end{align}
	and hence
	\begin{align}
		\sum_{s=0}^{\kappa_n(d)}
		e^{\theta_{n,d}s/n}
		=
		\frac{n}{|\theta_{n,d}|}
		\bigl(1+O(n^{-1})\bigr),
		\label{eq:geometric-sum-expansion}
	\end{align}
	uniformly over $d\in\mathcal D$.
	
	Taking the positive square root in
	\eqref{eq:tilted-variance-expansion} gives
	\begin{align}
		\sigma_{n,d}
		=
		n^{3/2}
		\sqrt{K''(\theta_d)}
		\bigl(1+O(n^{-1})\bigr),
		\label{eq:tilted-standard-deviation-expansion}
	\end{align}
	uniformly over $d\in\mathcal D$. Moreover,
	\eqref{eq:theta-nd-expansion}, together with the fact that
	$\{\theta_d:d\in\mathcal D\}$ is bounded away from zero, implies
	\begin{align}
		\frac{1}{|\theta_{n,d}|}
		=
		\frac{1}{|\theta_d|}
		\bigl(1+O(n^{-1})\bigr),
		\label{eq:inverse-saddlepoint-expansion}
	\end{align}
	uniformly over $d\in\mathcal D$.
	
	Substituting
	\eqref{eq:geometric-sum-expansion},
	\eqref{eq:tilted-standard-deviation-expansion},
	and \eqref{eq:inverse-saddlepoint-expansion} into
	\eqref{eq:lattice-contribution-geometric-reduction}, and using
	$\phi(0)=1/\sqrt{2\pi}$, yields
	\begin{align}
		T_{n,d}
		=
		\frac{1}{
			|\theta_d|
			\sqrt{2\pi nK''(\theta_d)}
		}
		\bigl(1+o(1)\bigr),
	\end{align}
	where the $o(1)$ term is uniform over $d\in\mathcal D$. This
	completes the proof.

	\subsection{Proof of Lemma~\ref{lem:near-identity-transform}}
	\label{app:proof-near-identity-transform}
	
	By Theorem~\ref{thm:uniform-strong-large-deviation},
	uniformly over all $t$ satisfying
	$\mu+\sigma t/\sqrt n\in\mathcal D$,
	\begin{align}
		&F_{\zeta_n}\!\left(
		n^2\left(
		\mu+\frac{\sigma t}{\sqrt n}
		\right)
		\right)
		=
		\frac{
			A\left(
			\mu+\frac{\sigma t}{\sqrt n}
			\right)
		}{
			\sqrt n
		}
		\exp\left\{
		-nI\left(
		\mu+\frac{\sigma t}{\sqrt n}
		\right)
		\right\}
		\left[
		1+
		\varrho_n\left(
		\mu+\frac{\sigma t}{\sqrt n}
		\right)
		\right].
		\label{eq:strong-ld-applied-to-transform}
	\end{align}
	Substitution into
	\eqref{eq:orb-rcu-deterministic-transform} gives
	\begin{align}
		\mathcal G_n(t)
		&=
		\frac{
			n\left[
			I(\mu)
			-
			I\left(
			\mu+\frac{\sigma t}{\sqrt n}
			\right)
			\right]
		}{
			\sqrt n\,|\theta_\mu|\sigma
		}
		+
		\frac{
			\ln A\left(
			\mu+\frac{\sigma t}{\sqrt n}
			\right)
			+
			\ln\left[
			1+
			\varrho_n\left(
			\mu+\frac{\sigma t}{\sqrt n}
			\right)
			\right]
		}{
			\sqrt n\,|\theta_\mu|\sigma
		}.
		\label{eq:transform-after-strong-ld}
	\end{align}
	By Taylor's theorem, for some point $\xi_{n,t}$ between
	$\mu$ and $\mu+\sigma t/\sqrt n$,
	\begin{align}
		I\left(
		\mu+\frac{\sigma t}{\sqrt n}
		\right)
		=
		I(\mu)
		+
		\theta_\mu
		\frac{\sigma t}{\sqrt n}
		+
		\frac{\sigma^2t^2}{2n}
		I''(\xi_{n,t}).
		\label{eq:rate-function-local-taylor}
	\end{align}
	Using $\theta_\mu=-|\theta_\mu|$, we obtain
	\begin{align}
		\mathcal G_n(t)-t
		&=
		-
		\frac{
			\sigma I''(\xi_{n,t})
		}{
			2|\theta_\mu|
		}
		\frac{t^2}{\sqrt n}
		+
		\frac{
			\ln A\left(
			\mu+\frac{\sigma t}{\sqrt n}
			\right)
			+
			\ln\left[
			1+
			\varrho_n\left(
			\mu+\frac{\sigma t}{\sqrt n}
			\right)
			\right]
		}{
			\sqrt n\,|\theta_\mu|\sigma
		}.
		\label{eq:transform-near-identity-decomposition}
	\end{align}
	The continuity of $I''(\cdot)$ and the continuity and strict
	positivity of $A(\cdot)$ on the compact interval $\mathcal D$,
	together with
	$\sup_{d\in\mathcal D}|\varrho_n(d)|\longrightarrow0$, imply
	that the coefficients on the right-hand side of
	\eqref{eq:transform-near-identity-decomposition} are uniformly
	bounded for all sufficiently large $n$. This proves
	\eqref{eq:near-identity-transform}.

	\subsection{Proof of Lemma~\ref{lem:gaussian-after-transform}}
	\label{app:proof-gaussian-after-transform}

	For brevity, define
	\begin{align}
		\mathsf T_n
		\triangleq
		\frac{
			\sqrt n
			\left(
			\mathsf D(
			\underline{\mathsf X},
			\underline{\mathsf Y}
			)
			-\mu
			\right)
		}{
			\sigma
		}.
		\label{eq:standardized-metric-proof}
	\end{align}
	By Theorem~\ref{thm:gauss of full D}, there exists a constant
	$0<C_1<\infty$ such that
	\begin{align}
		\sup_{x\in\mathbb R}
		\left|
		\Pr[\mathsf T_n\ge x]
		-
		Q(x)
		\right|
		\le
		\frac{C_1}{\sqrt n}
		\label{eq:BE-for-standardized-metric}
	\end{align}
	for all sufficiently large $n$.
	
	Let $0<C_0<\infty$ be the constant in Lemma~\ref{lem:near-identity-transform}. For
	$|t|\le\sqrt{\ln n}$, define
	\begin{align}
		\Delta_n(t)
		\triangleq
		\frac{4C_0(1+t^2)}{\sqrt n}.
		\label{eq:transform-displacement}
	\end{align}
	Uniformly over $|t|\le\sqrt{\ln n}$,
	\begin{align}
		\frac{
			\left|t\pm\Delta_n(t)\right|
		}{
			\sqrt n
		}
		\le
		\sqrt{\frac{\ln n}{n}}
		+
		\frac{4C_0(1+\ln n)}{n}
		\longrightarrow0.
		\label{eq:shifted-points-vanish}
	\end{align}
	Since \eqref{eq:d-compact-set} ensures that $\mu$ lies strictly inside $\mathcal D$, it follows that for all sufficiently large $n$,
	\begin{align}
		\mu+
		\frac{
			\sigma\left(t\pm\Delta_n(t)\right)
		}{
			\sqrt n
		}
		\in\mathcal D
		\label{eq:shifted-points-in-D}
	\end{align}
	uniformly over $|t|\le\sqrt{\ln n}$.
	
	Moreover,
	$\sup_{|t|\le\sqrt{\ln n}}\Delta_n(t)
	\le 4C_0(1+\ln n)/\sqrt n\to0$.
	Thus, for all sufficiently large $n$,
	\begin{align}
		1+
		\left(
		t\pm\Delta_n(t)
		\right)^2
		\le
		3(1+t^2)
		\label{eq:shifted-square-bound}
	\end{align}
	uniformly over $|t|\le\sqrt{\ln n}$. Applying Lemma~\ref{lem:near-identity-transform} at the two
	points in \eqref{eq:shifted-points-in-D}, we obtain
	\begin{align}
		\mathcal G_n\!\left(
		t-\Delta_n(t)
		\right)
		&\le
		t-\Delta_n(t)
		+
		\frac{3C_0(1+t^2)}{\sqrt n}
		<t,
		\nonumber\\
		\mathcal G_n\!\left(
		t+\Delta_n(t)
		\right)
		&\ge
		t+\Delta_n(t)
		-
		\frac{3C_0(1+t^2)}{\sqrt n}
		>t,
		\label{eq:transform-values-around-t}
	\end{align}
	uniformly over $|t|\le\sqrt{\ln n}$. The monotonicity of $\mathcal G_n(\cdot)$ and \eqref{eq:transform-values-around-t} imply that
	\begin{align}
		\left\{
		\mathsf T_n
		\ge
		t+\Delta_n(t)
		\right\}
		\subseteq
		\left\{
		\mathcal G_n(\mathsf T_n)
		\ge t
		\right\}
		\subseteq
		\left\{
		\mathsf T_n
		>
		t-\Delta_n(t)
		\right\}
		\subseteq
		\left\{
		\mathsf T_n
		\ge
		t-\Delta_n(t)
		\right\},
		\label{eq:transformed-event-sandwich}
	\end{align}
	uniformly over $|t|\le\sqrt{\ln n}$.
	
	Applying \eqref{eq:BE-for-standardized-metric} to
	\eqref{eq:transformed-event-sandwich} gives
	\begin{align}
		Q\!\left(
		t+\Delta_n(t)
		\right)
		-
		\frac{C_1}{\sqrt n}
		&\le
		\Pr\left[
		\mathcal G_n(\mathsf T_n)\ge t
		\right]
		\le
		Q\!\left(
		t-\Delta_n(t)
		\right)
		+
		\frac{C_1}{\sqrt n},
		\label{eq:transformed-probability-sandwich}
	\end{align}
	uniformly over $|t|\le\sqrt{\ln n}$.
	
	For any $s$ between $t$ and
	$t\pm\Delta_n(t)$,
	\begin{align}\label{eq:st}
		\ln\frac{\phi(s)}{\phi(t)}
		=
		\frac{t^2-s^2}{2}
		\le
		|t|\Delta_n(t)
		+
		\frac{\Delta_n^2(t)}{2}.
	\end{align}
	Uniformly over $|t|\le\sqrt{\ln n}$, the right-hand side of \eqref{eq:st} converges to zero. Hence, for all sufficiently large $n$, $\phi(s)\le2\phi(t)$. The mean-value theorem therefore gives
	\begin{align}
		\left|
		Q\!\left(t\pm\Delta_n(t)\right)-Q(t)
		\right|
		&\le
		2\Delta_n(t)\phi(t)
		\le
		\frac{C_2}{\sqrt n},
		\label{eq:shifted-Gaussian-tail}
	\end{align}
	uniformly over $|t|\le\sqrt{\ln n}$, for some constant
	$0<C_2<\infty$ independent of $n$ and $t$, where the last inequality
	follows from $\sup_{t\in\mathbb R}(1+t^2)\phi(t)<\infty$.

	Combining
	\eqref{eq:transformed-probability-sandwich} and
	\eqref{eq:shifted-Gaussian-tail}, we obtain
	\begin{align}
		Q(t)-\frac{C_1+C_2}{\sqrt n}
		&\le
		\Pr\left[
		\mathcal G_n(\mathsf T_n)\ge t
		\right]
		\le
		Q(t)+\frac{C_1+C_2}{\sqrt n},
	\end{align}
	uniformly over $|t|\le\sqrt{\ln n}$. Equivalently,
	\begin{align}
		\sup_{|t|\le\sqrt{\ln n}}
		\left|
		\Pr\left[
		\mathcal G_n(\mathsf T_n)\ge t
		\right]
		-
		Q(t)
		\right|
		\le
		\frac{C_1+C_2}{\sqrt n}.
		\label{eq:transformed-Gaussian-central-range}
	\end{align}
	
	It remains to consider $|t|>\sqrt{\ln n}$. For
	$t>\sqrt{\ln n}$, the monotonicity of tail probabilities and
	\eqref{eq:transformed-Gaussian-central-range} imply that
	\begin{align}
		\Pr\left[
		\mathcal G_n(\mathsf T_n)\ge t
		\right]
		&\le
		\Pr\left[
		\mathcal G_n(\mathsf T_n)
		\ge\sqrt{\ln n}
		\right]
		\le
		Q\!\left(\sqrt{\ln n}\right)
		+
		\frac{C_1+C_2}{\sqrt n}.
		\label{eq:right-tail-probability-bound}
	\end{align}
	In addition, the monotonicity of $Q(\cdot)$ gives $0\le Q(t)\le Q\!\left(\sqrt{\ln n}\right)$. Thus, both
	$\Pr[\mathcal G_n(\mathsf T_n)\ge t]$ and $Q(t)$ lie between
	zero and $Q(\sqrt{\ln n})+(C_1+C_2)/\sqrt n$. Consequently,
	\begin{align}
		\left|
		\Pr\left[
		\mathcal G_n(\mathsf T_n)\ge t
		\right]
		-
		Q(t)
		\right|
		&\le
		\max\left\{
		\Pr\left[
		\mathcal G_n(\mathsf T_n)\ge t
		\right],
		Q(t)
		\right\}
		\le
		Q\!\left(\sqrt{\ln n}\right)
		+
		\frac{C_1+C_2}{\sqrt n}.
		\label{eq:transformed-Gaussian-right-tail}
	\end{align}
	
	For $t<-\sqrt{\ln n}$, the monotonicity of tail probabilities
	gives
	\begin{align}
		\Pr\left[
		\mathcal G_n(\mathsf T_n)\ge t
		\right]
		&\ge
		\Pr\left[
		\mathcal G_n(\mathsf T_n)
		\ge-\sqrt{\ln n}
		\right]
		\ge
		Q\!\left(-\sqrt{\ln n}\right)
		-
		\frac{C_1+C_2}{\sqrt n},
		\label{eq:left-tail-probability-lower-bound}
	\end{align}
	where the last inequality follows from
	\eqref{eq:transformed-Gaussian-central-range}. Using the
	symmetry relation $Q(-x)=1-Q(x)$, we obtain
	\begin{align}
		0
		&\le
		1-
		\Pr\left[
		\mathcal G_n(\mathsf T_n)\ge t
		\right]
		\le
		Q\!\left(\sqrt{\ln n}\right)
		+
		\frac{C_1+C_2}{\sqrt n}.
		\label{eq:left-tail-complement-bound}
	\end{align}
	The monotonicity and symmetry of $Q(\cdot)$ also yield
	\begin{align*}
		0
		\le
		1-Q(t)
		\le
		1-Q\!\left(-\sqrt{\ln n}\right)
		=
		Q\!\left(\sqrt{\ln n}\right).
	\end{align*}
	Therefore, we obtain
	\begin{align}
		\left|
		\Pr\left[
		\mathcal G_n(\mathsf T_n)\ge t
		\right]
		-
		Q(t)
		\right|
		&=
		\left|
		\left(
		1-
		\Pr\left[
		\mathcal G_n(\mathsf T_n)\ge t
		\right]
		\right)
		-
		\left(
		1-Q(t)
		\right)
		\right|
		\nonumber\\
		&\le
		\max\left\{
		1-
		\Pr\left[
		\mathcal G_n(\mathsf T_n)\ge t
		\right],
		1-Q(t)
		\right\}
		\nonumber\\
		&\le
		Q\!\left(\sqrt{\ln n}\right)
		+
		\frac{C_1+C_2}{\sqrt n}.
		\label{eq:transformed-Gaussian-left-tail}
	\end{align}
	
	Finally, the standard Gaussian tail bound
	$Q(x)\le\phi(x)/x$, valid for every $x>0$, yields
	\begin{align}
		Q\!\left(\sqrt{\ln n}\right)
		&\le
		\frac{
			\phi\!\left(\sqrt{\ln n}\right)
		}{
			\sqrt{\ln n}
		}
		=
		\frac{1}{
			\sqrt{2\pi n\ln n}
		}
		=
		O\!\left(\frac{1}{\sqrt n}\right).
		\label{eq:growing-Gaussian-tail-bound}
	\end{align}
	Therefore, combining
	\eqref{eq:transformed-Gaussian-central-range},
	\eqref{eq:transformed-Gaussian-right-tail}, and
	\eqref{eq:transformed-Gaussian-left-tail}, there exists a
	constant $0<C<\infty$ such that
	\begin{align}
		\sup_{t\in\mathbb R}
		\left|
		\Pr\left[
		\mathcal G_n(\mathsf T_n)\ge t
		\right]
		-
		Q(t)
		\right|
		\le
		\frac{C}{\sqrt n}
		\label{eq:global-Gaussian-after-transform}
	\end{align}
	for all sufficiently large $n$. Recalling the definition of
	$\mathsf T_n$ in \eqref{eq:standardized-metric-proof},
	\eqref{eq:global-Gaussian-after-transform} is precisely
	\eqref{eq:gaussian-after-transform}, which completes the
	proof.

	\bibliographystyle{IEEEtran}
	\bibliography{ref.bib}
\end{document}